\documentclass[aps,prl,reprint,superscriptaddress,longbibliography]{revtex4-2}

\usepackage{amsmath,amssymb,amsthm,bm,mathtools}
\usepackage{graphicx}
\usepackage{xcolor}
\usepackage[colorlinks=true,allcolors=blue]{hyperref}

\newtheorem{theorem}{Theorem}
\newtheorem{proposition}{Proposition}

\newtheorem{lemma}{Lemma}
\newtheorem{corollary}{Corollary}
\theoremstyle{definition}

\theoremstyle{plain}

\newcommand{\DKL}{D_{\mathrm{KL}}}
\newcommand{\Dcoh}{D_{\mathrm{coh}}}
\newcommand{\Dloc}{D_{\mathrm{KL,loc}}^{\min}}
\newcommand{\Scross}{D_{\mathrm{cross}}^{(0)}}
\newcommand{\Scrossrho}{D_{\mathrm{cross}}^{(\rho)}}
\newcommand{\Snull}{\mathbf S_{\mathrm{null}}}
\newcommand{\Sprod}{\mathbf S_{\mathrm{prod}}}
\newcommand{\Strue}{\mathbf S_{\mathrm{true}}}
\newcommand{\Stwo}{\mathbf S}
\newcommand{\Pfun}{P}

\newcommand{\diag}{\mathrm{diag}}
\newcommand{\strictdiag}{\textit{diagonal null}}
\newcommand{\EPR}{\Phi}
\newcommand{\EPRtot}{\Phi_{\mathrm{total}}}

\begin{document}

\title{Cross-Spectral Witness for Hidden Nonequilibrium Beyond the Scalar Ceiling}

\author{Yuda Bi}
\email{ybi3@gsu.edu}
\affiliation{Translational Research in Neuroimaging and Data Science (TReNDS), Georgia State University, Atlanta, Georgia 30303, USA}
\author{Vince D. Calhoun}
\affiliation{Translational Research in Neuroimaging and Data Science (TReNDS), Georgia State University, Atlanta, Georgia 30303, USA}
\affiliation{School of Electrical and Computer Engineering, Georgia Institute of Technology, Atlanta, Georgia 30332, USA}

\date{\today}

\begin{abstract}
Partial observation is a pervasive obstacle in nonequilibrium physics: coarse graining may absorb hidden forcing into an apparently equilibrium-like reduced description, so a driven system can look reversible through the only variables one can measure.
For scalar Gaussian observables of linear stochastic systems, no time-irreversibility statistic can detect the underlying drive.
The Lucente--Crisanti ceiling constrains what one channel carries; what two channels carry is a different question, with a sharp closed-form answer.
Two simultaneously observed channels retain an off-diagonal cross-spectral sector inaccessible to any scalar reduction; under channel-separable multiplicative structure the observed-channel response factors cancel identically, leaving a closed-form cross-spectral witness controlled only by the hidden spectrum, the loadings, and the innovation scales, strictly positive at every nonzero cross-coupling including at exact timescale coalescence where every scalar reduction is blind.
Within general CSM this certifies shared hidden-sector drive; under the additional one-way coupling assumption the witness identifies the total entropy production rate at leading order with a square-root scaling.
\end{abstract}
\maketitle

\section{Introduction}

A fundamental problem haunts partial observation of driven systems. When only a subset of coordinates is monitored in a coupled stochastic system, a hidden external drive can be entirely absorbed into effective parameters of the visible channel---renormalized friction, renormalized noise amplitude---until the single-probe data looks indistinguishable from thermal equilibrium \cite{Seifert2012,KawaiParrondoVandenBroeck2007,RoldanParrondo2010,Esposito2012,MehlLanderSeifert2012,Maes2020}. The driven system is dissipating energy, but the observer fitting a one-channel effective model is blind to it.

A passive colloidal bead in an active bath satisfies the fluctuation--dissipation relation despite being driven \cite{FodorEtAl2016,BiskerEtAl2017}; in climate dynamics, subsurface thermocline forcing of sea-surface temperature is of order unity \cite{Hasselmann1976,FrankignoulHasselmann1977,PenlandSardeshmukh1995} yet invisible from the surface alone. Lucente et al.\ proved that no time-irreversibility statistic built from a single scalar Gaussian channel can detect the underlying drive \cite{LucenteEtAl2022}---the \emph{scalar Gaussian ceiling}. Yet the full system is far from equilibrium \cite{HorowitzGingrich2020,KawaguchiNakayama2013}. Where, then, is the nonequilibrium hiding?

In this Letter, we show that it hides in the \emph{off-diagonal cross-spectral sector} between two simultaneously observed channels---a sector that no scalar reduction can access. Under channel-separable multiplicative structure (CSM)---encompassing Mori--Zwanzig reductions, common-bath couplings, and multichannel linear response to shared latent forcing---the observed-channel response factors cancel identically. What survives is a closed-form cross-spectral witness, Eq.~\eqref{eq:rho-csm}, controlled solely by the hidden spectrum, the loadings, and the innovation scales, strictly positive wherever the hidden drive couples both channels---including at exact timescale coalescence where every single-channel statistic is provably blind. Under one-way coupling it identifies the total entropy production rate at leading order. The scalar ceiling is not revoked by a more powerful scalar statistic---it is bypassed by changing the geometry of the measurement.

\section{Hidden Drive Through Two Observed Channels}

\emph{Physical setting.}---Consider two probes---two laser-tracked colloidal beads in a shared active bath, or two climate indices driven by a common subsurface mode---each governed by its own local relaxation dynamics and both subject to a shared persistent forcing that evolves independently of the probes. The observer records only the probe positions; the forcing mode is unobserved. This is the discrete-time analogue of a Mori--Zwanzig reduction \cite{Zwanzig1961,Mori1965,ChorinHaldKupferman2000} with a finite-rank retained latent memory kernel. The single-mode benchmark, in which one hidden persistent mode drives both channels through a unit-norm loading vector $u_1^2+u_2^2=1$, makes the mechanism closed-form; the general channel-separable latent sector is introduced in Eq.~\eqref{eq:csm-general} below. The benchmark equations of motion are:
\begin{equation}
\begin{aligned}
X_{t+1}^{(1)} &= a_1 X_t^{(1)} + \lambda u_1 F_t + \epsilon_{t+1}^{(1)},\\
X_{t+1}^{(2)} &= a_2 X_t^{(2)} + \lambda u_2 F_t + \epsilon_{t+1}^{(2)},
\end{aligned}
\label{eq:model-x}
\end{equation}
\begin{equation}
\begin{aligned}
F_{t+1} &= bF_t+\eta_{t+1},\\
u_1^2+u_2^2 &= 1.
\end{aligned}
\label{eq:model-f}
\end{equation}
Here $|a_1|,|a_2|,|b|<1$; the innovations $\epsilon_t^{(i)}\sim\mathcal N(0,\sigma_{\epsilon_i}^2)$ and $\eta_t\sim\mathcal N(0,\sigma_\eta^2)$ are mutually independent; and only $(X_t^{(1)},X_t^{(2)})$ are observed. The coupling $\lambda$ controls the strength of the shared drive. In the benchmark, the hidden mode $F_t$ evolves independently of the observed channels (\emph{one-way coupling}), as a bath or external drive would. The filter-free witness derived below requires only CSM, not one-way coupling; one-way coupling is an additional assumption needed solely for the thermodynamic bridge of Sec.~\ref{sec:thermo}. Write
\[
\Pfun_c(\omega)=|1-ce^{-i\omega}|^2=1+c^2-2c\cos\omega.
\]
Define
\begin{equation}
\begin{aligned}
\mathbf D_\epsilon(\omega)
&=
\diag\!\left(
\frac{\sigma_{\epsilon_1}^2}{\Pfun_{a_1}(\omega)},
\frac{\sigma_{\epsilon_2}^2}{\Pfun_{a_2}(\omega)}
\right),\\
\mathbf h(\omega)
&=
\begin{pmatrix}
\dfrac{u_1}{1-a_1 e^{-i\omega}}\\[4pt]
\dfrac{u_2}{1-a_2 e^{-i\omega}}
\end{pmatrix}.
\end{aligned}
\label{eq:Deps-h}
\end{equation}
Then the exact observed spectral matrix can be written compactly as
\begin{equation}
\mathbf S_{\mathrm{true}}(\omega)
=
\mathbf D_\epsilon(\omega)
+
\frac{\lambda^2\sigma_\eta^2}{\Pfun_b(\omega)}
\mathbf h(\omega)\mathbf h(\omega)^\ast .
\label{eq:true-matrix}
\end{equation}
The compact form~\eqref{eq:true-matrix} is the rank-one instance of a general structure. For any causal stable linear filter $H_i(z)$ governing channel~$i$, with innovation spectrum $\sigma_{\epsilon_i}^2|H_i(e^{-i\omega})|^2$ and finite-rank shared latent spectrum $K_F(\omega)=L\,S_Z(\omega)L^\ast$ entering each channel at the same dynamical point as the innovation, the observed spectral matrix takes the channel-separable multiplicative form
\begin{equation}
\Strue(\omega)=\mathbf H(\omega)\bigl[\Sigma_\epsilon+K_F(\omega)\bigr]\mathbf H(\omega)^\ast,
\label{eq:csm-general}
\end{equation}
where $\mathbf H(\omega)=\diag(H_1(e^{-i\omega}),H_2(e^{-i\omega}))$, $\Sigma_\epsilon=\diag(\sigma_{\epsilon_1}^2,\sigma_{\epsilon_2}^2)$, $L\in\mathbb{C}^{2\times r}$ is a fixed loading matrix, and $S_Z(\omega)$ is an $r\times r$ Hermitian positive-semidefinite latent spectral density. CSM is the multichannel analogue of a local linear response per channel dressed by a shared finite-dimensional latent spectral sector; it arises from Mori--Zwanzig reductions with a finite-rank retained latent sector~\cite{Zwanzig1961,GivonKupfermanStuart2004}, common-mode bath coupling, and multichannel linear response to a shared latent forcing bank. Equation~\eqref{eq:true-matrix} is the rank-one benchmark $r=1$, $L=\lambda u$, $S_Z=S_F=\sigma_\eta^2/\Pfun_b$, used throughout for closed forms (componentwise expansion in Appendix~B).

CSM is broad but not universal: it requires the hidden drive to enter each channel at the same dynamical point as the channel's own innovation noise, with fixed frequency-independent loadings and linear response. Each probe's local dynamics shapes both its own thermal fluctuations and the externally imposed drive in the same way---a colloidal bead in a harmonic trap responds to both thermal kicks and external forcing through the same overdamped transfer function. CSM \emph{fails} when the latent drive enters through a different pathway than the innovation (delayed or band-shifted couplings), when loadings are state-dependent, or when channel responses are strongly nonlinear. CSM is a sharp class, and the filter cancellation it implies is exact. In practice, CSM can be tested empirically: the filter-free prediction requires $\Dcoh$ to be invariant when different AR($p$) models are fitted to the observed channels. If the witness varies with the assumed channel order, CSM is violated for the data at hand.

\emph{The null hypothesis: independent probes.}---The natural baseline for ``no shared hidden forcing'' is two independent probes, each described by its own effective AR(1) model---the \strictdiag:
\begin{equation}
\Snull(\omega;\theta)=\diag\!\left(\frac{\tilde\sigma_1^2}{\Pfun_{\tilde a_1}(\omega)},\;\frac{\tilde\sigma_2^2}{\Pfun_{\tilde a_2}(\omega)}\right),
\label{eq:strict-diag-null}
\end{equation}
with $\theta=(\tilde a_1,\tilde\sigma_1^2,\tilde a_2,\tilde\sigma_2^2)$.
This is the spectral matrix of two probes in separate thermal baths with no coupling between them. Any cross-channel structure in the data that this family cannot fit is evidence of shared input. Enriching the null with off-diagonal structure concedes common input at the null level \cite{Brillinger2001,Geweke1982,Geweke1984} and is treated in the enrichment-boundary analysis below.

Alongside the parametric null~\eqref{eq:strict-diag-null}, the observed joint spectrum has an exact scalar-marginal reduction
\begin{equation}
\Sprod(\omega):=
\diag\!\left(
S_{11}^{\mathrm{true}}(\omega),
S_{22}^{\mathrm{true}}(\omega)
\right),
\label{eq:diag-prod-null}
\end{equation}
which preserves each channel's power spectrum exactly while removing all joint structure---what a pair of independent single-channel observers would infer. The key point below is that the residual's spectral profile is intrinsic to the shared latent sector, carrying no trace of either probe's local dynamics.

\section{The Hidden Sector Survives Scalar Marginalization}

\emph{What the scalar marginals cannot see.}---Two simultaneously observed channels carry information that neither marginal alone possesses: the joint dependence between them, encoded in the normalized cross-coherence $\rho^2(\omega):=|S_{12}^{\mathrm{true}}(\omega)|^2/[S_{11}^{\mathrm{true}}(\omega)S_{22}^{\mathrm{true}}(\omega)]\in[0,1]$. The information gap to the scalar marginals is the single real-valued observable
\begin{equation}
\Dcoh := -\frac{1}{4\pi}\int_{-\pi}^{\pi}\log\!\bigl(1-\rho^2(\omega)\bigr)\,d\omega,
\label{eq:exact-coh}
\end{equation}
the classical coherence integral of Gaussian stationary processes \cite{Geweke1982,Komaee2020}; we call it the \emph{cross-spectral witness}. The physical content of this Letter is not that $\Dcoh$ exists---that is a classical statement about correlated Gaussian processes---but what its spectral profile becomes once the hidden drive enters through a channel-separable shared latent sector.

\emph{Why the observed-channel filters cancel.}---The reason is geometric. Under CSM~\eqref{eq:csm-general}, the cross-spectrum $S_{12}^{\mathrm{true}}$ carries a factor $H_1(\omega)\overline{H_2(\omega)}$---the product of one channel's causal filter with the complex conjugate of the other. Each scalar marginal carries $|H_i(\omega)|^2$. In the ratio $\rho^2=|S_{12}|^2/(S_{11}S_{22})$, the filter moduli $|H_1|^2|H_2|^2$ appear identically in numerator and denominator and cancel pointwise in frequency:
\begin{equation}
\rho^2(\omega)=\frac{|K_{12}(\omega)|^2}{\bigl(\sigma_{\epsilon_1}^2+K_{11}(\omega)\bigr)\bigl(\sigma_{\epsilon_2}^2+K_{22}(\omega)\bigr)}.
\label{eq:rho-csm}
\end{equation}
The witness is intrinsic to the hidden sector: only the latent-sector entries $K_{ij}$ and the innovation scales $\sigma_{\epsilon_i}^2$ survive. This is the spectral analogue of common-mode rejection in multichannel signal processing \cite{Brillinger2001}: each channel's local response acts as a common factor that divides out when the two channels are compared, leaving only the shared hidden content. Hermitian positive-semidefiniteness of $K_F$ ensures $\rho^2\in[0,1)$ pointwise (via Cauchy--Schwarz), and the witness is strictly positive whenever $K_{12}\ne 0$ on a set of positive measure---regardless of the observed channel poles, including at exact coalescence between observed and hidden timescales. Derivation and uniform validity across finite rank are in Appendix~C.

Three consequences follow: false equilibrium appearances are a projection artifact of single-channel observation, not an intrinsic property; enlarging the measurement from one to two channels changes the retained information class qualitatively; and the surviving invariant is controlled solely by the hidden spectrum, absent from any scalar reduction. Within general CSM, Eq.~\eqref{eq:rho-csm} certifies a shared hidden-sector drive; the thermodynamic identification of that drive as entropy production requires the further specialization to one-way coupling (Sec.~\ref{sec:thermo}).

In the single-mode benchmark $K_F=\lambda^2 S_F\,uu^\top$ of Eq.~\eqref{eq:true-matrix}, Eq.~\eqref{eq:rho-csm} reduces to
\begin{equation}
\rho^2(\omega)=\frac{\lambda^4 u_1^2 u_2^2\,S_F(\omega)^2}{\bigl(\sigma_{\epsilon_1}^2+\lambda^2 u_1^2 S_F(\omega)\bigr)\bigl(\sigma_{\epsilon_2}^2+\lambda^2 u_2^2 S_F(\omega)\bigr)},
\label{eq:rho-exact}
\end{equation}
the closed form used throughout for explicit calculations. Physically, $\rho^2(\omega)$ measures the fraction of joint spectral power at frequency $\omega$ attributable to the shared hidden mode: at $\lambda=0$ the probes are uncorrelated ($\rho^2\equiv 0$); at large $\lambda$ the hidden mode dominates ($\rho^2\to 1$); in between, the spectral shape traces $S_F(\omega)$. This rank-one case exhibits the mechanism through two sharp fingerprints---a universal quartic cross law and a coalescence singularity that cross spectra remove---which we now develop as physical content reconnecting the witness to the scalar-impossibility literature.

\section{Rank-One Benchmark: Quartic Law, Coalescence, and Enrichment}

The rank-one closed form~\eqref{eq:rho-exact} is exact at all coupling strengths. We now expand it at weak coupling $\lambda\to 0$, where the quartic structure of the cross-spectral witness becomes explicit and reveals why the scalar ceiling is geometric rather than fundamental. \emph{The filter-free witness~\eqref{eq:rho-csm} and its strict positivity are proven for all finite-rank CSM (Appendix~C, Theorem~1). The quartic closed forms below are proven for rank one; the same qualitative structure---quartic onset, coalescence survival, filter-free cross coefficient---persists at higher ranks and is verified numerically to machine precision in Appendix~L.}

\emph{Quartic decomposition.}---At weak coupling $\lambda\to 0$, the spectral residual along the best-fit diagonal branch decomposes into two channelwise auto contributions---one from each probe's marginal mismatch---and a joint cross contribution from the off-diagonal sector (Appendices~C and~D):
\begin{equation}
\Dloc(\lambda)=D_{\mathrm{auto},1}^{\min}+D_{\mathrm{auto},2}^{\min}+\Scross(\lambda)+O(\lambda^6),
\label{eq:decomposition-structural}
\end{equation}
with $\Scross(\theta^\star(\lambda),\lambda)=\Scross(\theta_0,\lambda)+O(\lambda^6)$, so the cross contribution is invariant under diagonal reparametrization at quartic order. This decomposition makes the detection hierarchy explicit: a single-channel observer can always adjust its effective parameters $(\tilde a_i, \tilde\sigma_i^2)$ to absorb part of the hidden perturbation into each marginal spectrum. But no channel-by-channel adjustment can reach the off-diagonal sector; absorbing the cross contribution requires enlarging the reduced description itself.

\emph{Filter cancellation and the quartic cross law.}---At the null point, the normalized cross-coherence collapses to a purely latent object. For any causal stable filter $H_i(z)$ governing channel $i$, the observed-channel factors cancel exactly (Appendix~E):
\begin{equation}
\frac{|S_{12}^{\mathrm{true}}(\omega)|^2}{S_{11}^{0}(\omega)S_{22}^{0}(\omega)}
=\frac{\lambda^4 u_1^2 u_2^2}{\sigma_{\epsilon_1}^2\sigma_{\epsilon_2}^2}\,S_F(\omega)^2.
\label{eq:cancellation}
\end{equation}
The cross-spectral fingerprint of the hidden mode is insensitive to the local dynamics of each probe: whether the probes are overdamped Brownian particles, resonant oscillators, or higher-order systems, the ratio that enters the witness sees only the shared forcing. The leading cross contribution inherits this independence:
\begin{equation}
\Scross(\lambda)=\frac{\lambda^4 u_1^2 u_2^2}{\sigma_{\epsilon_1}^2\sigma_{\epsilon_2}^2}\,\mathcal I_F+O(\lambda^6),
\label{eq:cross-general}
\end{equation}
with $\mathcal I_F:=(4\pi)^{-1}\int_{-\pi}^{\pi}S_F(\omega)^2\,d\omega$ depending only on the hidden spectral density. For the AR(1) hidden mode, $C_{\mathrm{cross}}=u_1^2 u_2^2\sigma_\eta^4(1+b^2)/[2\sigma_{\epsilon_1}^2\sigma_{\epsilon_2}^2(1-b^2)^3]$. This observed-dynamics independence does not follow from the Geweke decomposition alone \cite{Geweke1982,Geweke1984}: it is specific to the scalar-marginal reduction inside CSM, where the filters act as a common multiplicative factor that divides out.

\emph{Coalescence singularity removal.}---Coalescence means the observed probes and the hidden driver oscillate at the same characteristic timescale ($a_i=b$). For a single-channel observer this is the worst case: the hidden mode's spectral shape becomes indistinguishable from the probe's own response, and the scalar detection coefficient vanishes through the factor $(a_i-b)^2$. Each auto contribution inherits this singularity (Appendix~A):
\begin{equation}
C_{\mathrm{auto}}^{(i)}=\frac{u_i^4\sigma_\eta^4}{2\sigma_{\epsilon_i}^4}\frac{b^2(a_i-b)^2}{(1-b^2)^3(1-a_ib)^2},\qquad i=1,2,
\label{eq:auto-coeff}
\end{equation}
so the full quartic law reads
\begin{equation}
\Dloc(\lambda)=\bigl(C_{\mathrm{auto}}^{(1)}+C_{\mathrm{auto}}^{(2)}+C_{\mathrm{cross}}\bigr)\lambda^4+O(\lambda^6).
\label{eq:full-quartic}
\end{equation}
At exact coalescence $a_1=a_2=b$, both auto coefficients vanish, but the cross coefficient $C_{\mathrm{cross}}$---which depends only on the hidden spectrum, not on the observed poles---stays intact:
\begin{equation}
\Dloc(\lambda)=C_{\mathrm{cross}}\lambda^4+O(\lambda^6)>0\quad\text{whenever } u_1 u_2\ne 0.
\label{eq:coalescence-positive}
\end{equation}
This is the singularity removal visible in Fig.~\ref{fig:main-geometry}: the projection that erases the hidden mode from each probe's marginal spectrum does not erase it from the joint spectrum. A single probe cannot tell whether its fluctuations come from its own thermal noise or from the hidden driver when both have the same correlation time, but \emph{two} probes can, because the hidden driver correlates them in a way that independent thermal noise cannot.

\begin{figure*}[t]
  \centering
  \includegraphics[width=\textwidth]{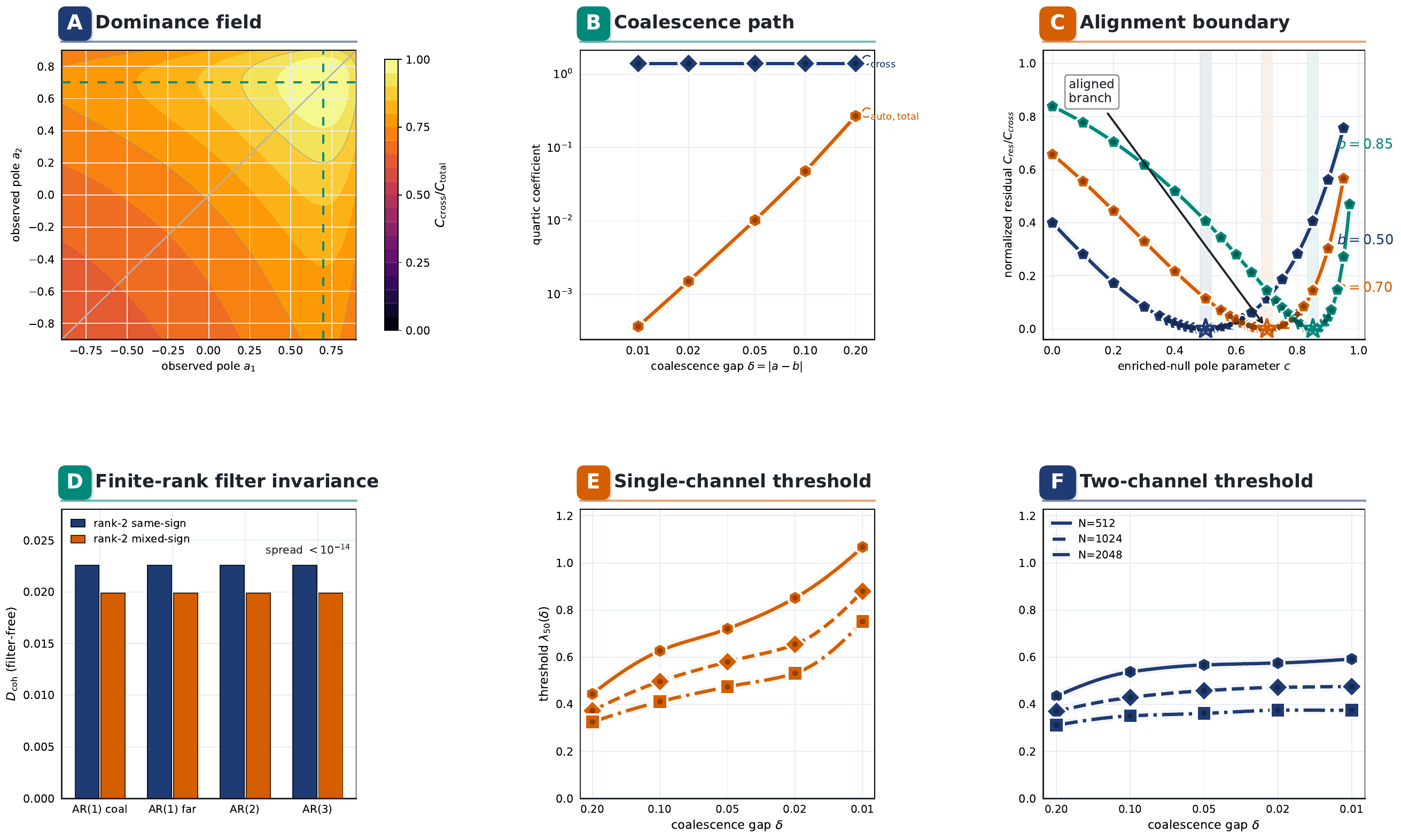}
  \caption{\textbf{The scalar ceiling is a projection singularity, and cross spectra cross it.} \textbf{A,} Fractional dominance $C_{\mathrm{cross}}/C_{\mathrm{total}}$ across the observed-pole plane $(a_1,a_2)$ for the rank-one benchmark; near coalescence $a_1=a_2=b$ (dashed lines, $b=0.7$) the cross-spectral witness carries the full residual. \textbf{B,} Along the coalescence path $a_1=a_2=b+\delta$, the auto contribution collapses while the witness remains finite. \textbf{C,} Enriched-null absorption: only a family aligned with the hidden sector ($c=b$) reabsorbs the witness, for $b=0.5,0.7,0.85$. \textbf{D,} Finite-rank filter invariance: $\Dcoh$ is identical across AR(1), AR(2), and AR(3) observed filters for rank-$2$ latent sectors (same-sign blue, mixed-sign orange); spread $<10^{-14}$ verifies Eq.~\eqref{eq:rho-csm}. \textbf{E,F,} Minimum coupling $\lambda_{50}(\delta)$ to reveal the hidden drive in single-channel (E) and two-channel (F) observation, $N=512,1024,2048$. The single-channel threshold diverges as $\delta\to 0$; the two-channel threshold stays bounded---the numerical imprint of the singularity removal. Symmetric loadings $u_1=u_2=1/\sqrt 2$ in E,F; supplementary figures cover asymmetric cases. Identification protocol in Appendix~H.}
  \label{fig:main-geometry}
\end{figure*}

\emph{Enrichment boundary: alignment, not flexibility.}---\label{sec:boundary}Can a richer reduced family reabsorb the witness at coalescence? The hidden drive contributes a single spectral direction to the cross sector. A richer fit can erase the part it covers, but no more. The integrated cross residual is the squared mismatch
\begin{equation}
\Scrossrho=\|q-\Pi_{\mathcal E_\rho}q\|^2,\qquad 0\le\Scrossrho\le\Scross,
\label{eq:enriched-span}
\end{equation}
where $q(\omega):=S_{12}^{\mathrm{true}}(\omega)/\sqrt{S_{11}^0(\omega)S_{22}^0(\omega)}$ is the normalized cross-spectral direction and $\Pi_{\mathcal E_\rho}$ projects onto the enriched family's cross-spectral span. Exact removal, $\Scrossrho=0$, requires the enriched family's cross-spectral direction to coincide with the hidden spectral shape---alignment, not flexibility (Appendix~G). Figure~\ref{fig:main-geometry}C confirms that only when the enriched parameter $c$ matches the hidden persistence $b$ does the residual vanish \cite{Brillinger2001,BresslerSeth2011,BastosSchoffelen2016,Friston2011}; in every other direction of enrichment, the witness remains strictly positive. This converts the question ``can a richer effective model absorb the nonequilibrium signal?'' into a question about spectral alignment with the hidden mode---a physical constraint, not a fitting limitation.

\section{Numerical Verification and Robustness}

\emph{Singularity removal at finite sample size.}---Panels~E and~F of Fig.~\ref{fig:main-geometry} show the singularity removal in simulated data at finite trajectory length ($N\in\{512,1024,2048\}$; identification protocol in Appendix~H). As the coalescence gap $\delta=a_i-b\to 0$, the minimum coupling $\lambda_{50}^{\mathrm{single}}(\delta)$ needed for single-channel detection diverges: the hidden mode becomes spectroscopically dark to one probe. The two-channel threshold $\lambda_{50}^{\mathrm{two}}(\delta)$, in contrast, stays bounded and approaches the exact prediction from above as $N$ grows. The hidden drive remains visible to two probes at couplings where single-probe identification has already failed. This split is robust across asymmetric loadings (Supplementary Figs.~\ref{fig:supp-r50}--\ref{fig:supp-asymmetric}).

\emph{Structural robustness.}---\label{sec:robustness}The filter-free cancellation is a structural property of the physical setup---a shared latent sector coupling two channels through their own local responses---not an artifact of the AR(1) benchmark. We verify this along three physical axes (Appendix~I, Fig.~\ref{fig:robustness}; Appendix~L). (i)~\emph{Bidirectional coupling}: adding feedback from the probes to the latent sector, with feedback strength up to $\mu/\gamma_f=0.5$, preserves the structural witness; the witness tracks shared-sector coupling regardless of directionality. (The thermodynamic bridge of Eq.~\eqref{eq:epr-sqrt} is specific to the one-way subclass and is not guaranteed beyond it.) (ii)~\emph{Higher-order probe dynamics}: replacing the AR(1) single-pole probe response with AR(2) and AR(3) multi-pole responses leaves $\Dcoh$ invariant at the $10^{-14}$ level---the filters cancel to machine precision. (iii)~\emph{Nonlinear channels}: adding cubic damping ($\kappa\le 0.015$) does not produce spurious identifications, and a witness-guided test built directly from~\eqref{eq:rho-csm} retains its discriminating power without refitting any hidden-driver model. Negative controls with $K_{12}\equiv 0$ yield $\Dcoh=0$ across all tested classes (Appendix~L), confirming that the witness reflects shared hidden drive rather than spurious cross-channel correlation.

\section{From Witness to Entropy Production}
\label{sec:thermo}

\emph{One-way coupling: when common drive is dissipation.}---The structural witness of Eq.~\eqref{eq:rho-csm} certifies a shared hidden drive under general CSM, with no assumption on the direction of coupling. When the latent process is dynamically autonomous of the probes---the system is \emph{one-way coupled}---this common drive acquires thermodynamic content: it is the sole source of entropy production in the joint system, because any coupling that breaks detailed balance must pass through the hidden sector. One-way coupling is not a technical limitation but a physical identification condition: it holds whenever the hidden sector is macroscopically large compared to the probes, so that the probes cannot appreciably back-react on the drive. This is the natural setting for a thermal bath driving colloidal probes at a different temperature, slow climate modes forcing surface variables, or neuromodulators driving cortical circuits.

To connect to the standard entropy production formalism we pass to the continuous-time limit of the benchmark: the Ornstein--Uhlenbeck process $dX_i=-\gamma_i X_i\,dt+\lambda u_i F\,dt+\sqrt{2D_i}\,dW_i$ for $i=1,2$, $dF=-\gamma_f F\,dt+\sqrt{2D_f}\,dW_f$, where the hidden mode $F$ evolves independently of the probes. The steady-state entropy production rate of the joint $(X_1,X_2,F)$ process is exactly quadratic in the coupling (Appendix~K):
\begin{equation}
\EPRtot=\alpha_2\,\lambda^2,
\label{eq:epr-exact}
\end{equation}
with $\alpha_2=u_1^2 D_f/[D_1(\gamma_1+\gamma_f)]+u_2^2 D_f/[D_2(\gamma_2+\gamma_f)]>0$. Physically, $\EPRtot$ measures the rate at which the hidden forcing degrades free energy by pushing the probes away from their equilibrium configurations; it is proportional to $\lambda^2$ because entropy production is a current $\times$ force quantity, and both current and force scale as $\lambda$ at leading order. Combined with $\Scross=C_{\mathrm{cross}}\lambda^4+O(\lambda^6)$, the bridge between the observable and the thermodynamic invariant is
\begin{equation}
\EPRtot=\frac{\alpha_2}{\sqrt{C_{\mathrm{cross}}}}\,\sqrt{\Scross}+O(\lambda^4).
\label{eq:epr-sqrt}
\end{equation}
The square-root scaling reflects the order mismatch: entropy production is quadratic in $\lambda$, the witness quartic, and the bridge between them is necessarily a square root---making $\Scross$ a quantitative dissipation probe, not merely a qualitative detector. Since the single-channel marginal entropy production is identically zero \cite{LucenteEtAl2022}, the dissipative content of the system passes entirely through the cross-spectral witness at leading order---precisely the quantity that the scalar ceiling forbids any single probe from seeing. At strong coupling the leading-order bridge becomes inaccurate, but the full nonperturbative witness $\Dcoh$ of Eq.~\eqref{eq:exact-coh} remains well-defined and strictly positive for all $\lambda>0$.

\section{Illustration on NOAA Climate Indices}
\label{sec:climate}

As a proof of concept on real data, we apply the witness to two pairs of NOAA monthly climate indices ($N\approx 900$), calibrated against 999 phase-randomized surrogates that destroy cross-spectral structure while preserving each channel's marginal spectrum. For Ni\~no~3.4 sea-surface temperature anomaly versus the Southern Oscillation Index---a tightly coupled pair within the ENSO basin---the witness lies $18.8\sigma$ above the null mean ($p\le 10^{-3}$), consistent with the strong atmosphere--ocean coupling known from the Walker circulation. For Ni\~no~3.4 versus the Atlantic Multidecadal Oscillation---an inter-basin teleconnection---the witness is $3.5\sigma$ above null, with $\rho^2(\omega)$ peaked in the decadal band (Fig.~\ref{fig:climate}), consistent with known low-frequency inter-basin coupling. Both computations use only Eq.~\eqref{eq:rho-csm} and the phase-randomized null, with no parametric hidden-driver model; we make no causal or directional claim from either pair. As a consistency check, smoothed-periodogram, Welch, and multitaper estimates of $\hat D_{\mathrm{coh}}$ agree within $7\%$ on the Ni\~no~3.4 vs.\ SOI pair, confirming that the witness is a stable spectral observable across estimation methods. The phase-randomized surrogate distribution provides a nonparametric confidence assessment: rejection at $p\le 10^{-3}$ with $18.8\sigma$ separation leaves no ambiguity about statistical significance.

\begin{figure*}[t]
  \centering
  \includegraphics[width=\textwidth]{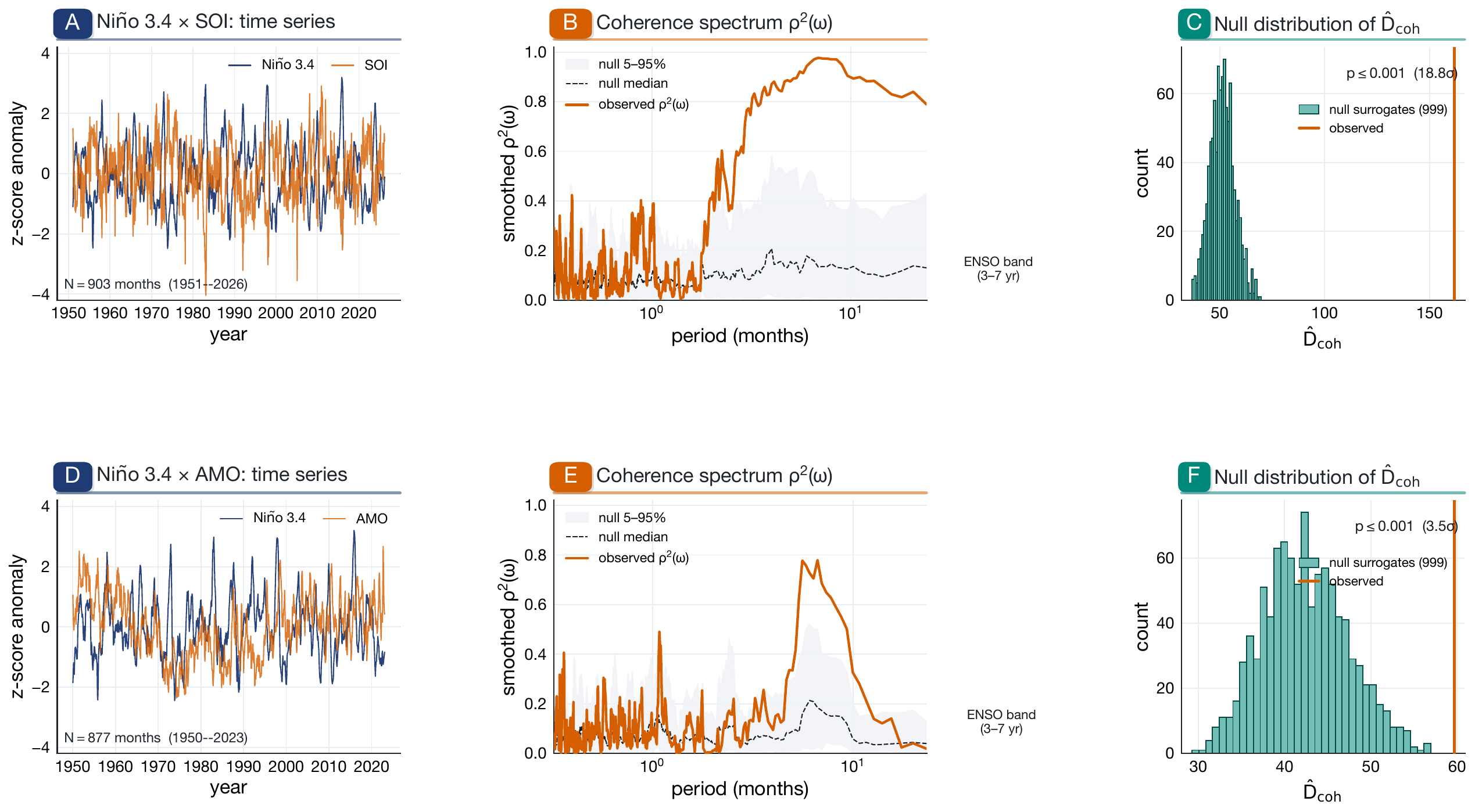}
  \caption{\textbf{Cross-spectral witness on real climate data.} \textbf{A,} Z-scored monthly time series of Niño~3.4 SST anomaly and the Southern Oscillation Index (1951--2026). \textbf{B,} Smoothed coherence $\rho^2(\omega)$ (red) versus the 5--95\% envelope of 999 phase-randomized surrogates (grey); ENSO band (3--7~yr) shaded. \textbf{C,} Null distribution of $\hat D_{\mathrm{coh}}$ (grey) with the observed value (red), $18.8\sigma$ above null. \textbf{D,E,F,} Same for Niño~3.4 vs.\ the Atlantic Multidecadal Oscillation (1950--2023), $3.5\sigma$. Panel E shows a structured decadal peak where Niño~3.4 and the Atlantic basin share maximum coherence, consistent with known inter-basin teleconnection but demonstrating correlation rather than direction of causation. No parametric hidden-driver model is fit in either case.}
  \label{fig:climate}
\end{figure*}

\section{Discussion}

We have shown that two simultaneously observed channels of a linear stochastic system carry a filter-free cross-spectral witness for the shared hidden drive, controlled solely by the hidden spectrum, loadings, and innovation scales. The witness is strictly positive wherever the hidden drive couples both channels---including at exact timescale coalescence where every single-channel statistic is provably blind---and under one-way coupling it identifies the total entropy production rate at leading order. The scalar Gaussian ceiling \cite{LucenteEtAl2022,CrisantiPuglisiVillamaina2012} is thus not a fundamental invisibility of the nonequilibrium physics but a projection artifact of the one-channel measurement geometry: the nonequilibrium is there, in the off-diagonal cross-spectral sector, waiting for a second probe to reveal it.

\emph{Experimental pathways.}---Dual-probe colloidal experiments with two same-size beads in a shared optical trap or active bath satisfy CSM by construction \cite{Fodor2016,BattleEtAl2016,Bechinger2016}: each bead's local response is its own harmonic restoring force, and the shared bath mode enters both channels at the same dynamical point. The witness is directly measurable from synchronized position recordings via standard spectral estimation, requiring no microscopic model fit. Multichannel neural recordings with a global latent input \cite{SekizawaItoOizumi2024} plausibly satisfy CSM for nearby same-region electrode pairs, while violating it across cortical layers with distinct synaptic time constants; the witness could discriminate these two regimes.

\emph{Relation to the Geweke decomposition.}---The coherence integral $\Dcoh$ is precisely Geweke's measure of contemporaneous linear dependence \cite{Geweke1982,Geweke1984}. Geweke, however, treats it as a descriptive statistic of two observed processes, with no latent-variable decomposition: in Geweke's framework $\rho^2(\omega)$ is simply a ratio of observable spectra, and there is no notion of a hidden sector whose fingerprint survives marginalization. Under CSM the same quantity acquires structural content: its spectral profile is intrinsic to the hidden sector, independent of the observed-channel dynamics, because the filter moduli cancel identically in the ratio. Neither the coalescence singularity removal nor the EPR bridge has an analog in Geweke's framework, since both require the CSM decomposition into channel filters and latent spectrum. Figure~\ref{fig:main-geometry}D confirms the filter-free property numerically: $\Dcoh$ is identical across AR(1), AR(2), and AR(3) observed filters to $10^{-14}$ precision.

\emph{Unification.}---The scalar darkness of Ref.~\cite{BiZhangCalhoun2026scalar} and the cross-spectral witness of the present Letter are two faces of the same geometry: the former is where the off-diagonal anomaly vanishes by projection ($|V|=1$), the latter is where it first turns on ($|V|=2$). The quartic onset at weak coupling, the coalescence singularity in the auto sector, and the filter-free cross sector are all manifestations of a single tangent-space structure: the reduced observer's parameter space can absorb the hidden perturbation along its own tangent directions but not across the off-diagonal gap. A full $N$-channel theory extending pairwise application of the witness with aggregation across dominant directions is a natural next step.

\begin{acknowledgments}
[Acknowledgments to be added.]
\end{acknowledgments}

All data and code needed to reproduce the results are available at \url{https://github.com/yudabitrends/cross-spectral-detectability} upon publication. Climate index data are publicly available from NOAA.

\bibliography{multivariate_followup_refs}

\clearpage
\onecolumngrid
\appendix
\setcounter{equation}{0}
\setcounter{figure}{0}
\setcounter{table}{0}
\renewcommand{\theequation}{\Alph{section}\arabic{equation}}
\makeatletter\@addtoreset{equation}{section}\makeatother
\renewcommand{\thefigure}{S\arabic{figure}}
\renewcommand{\thetable}{S\arabic{table}}

\begin{center}
{\Large\bfseries Supplemental Material}\\[0.35em]
{\large\bfseries Cross-Spectral Witness for Hidden Nonequilibrium Beyond the Scalar Ceiling}
\end{center}

\section*{Appendix Contents}

\begin{itemize}
  \item Appendix A: Scalar quartic-law foundation
  \item Appendix B: Exact multivariate spectrum and cross-spectrum lemmas
  \item Appendix C: Multivariate Whittle/KL decomposition and Hermitian log-det expansion
  \item Appendix D: Local diagonal branch and absorption boundary
  \item Appendix E: Cancellation identity and cross coefficient closed form
  \item Appendix F: Scalar-to-multivariate inheritance of the auto terms
  \item Appendix G: Boundary characterization for enriched nulls
  \item Appendix H: Symbolic verification and finite-sample records
  \item Appendix I: Robustness experiment protocols
  \item Appendix J: Related works and scope
  \item Appendix K: Benchmark-specific thermodynamic extension
  \item Appendix L: Finite-rank latent-sector verification (Route~1)
\end{itemize}

\section{Scalar Quartic-Law Foundation}

This appendix records the scalar results that the present multivariate analysis requires: the exact relative perturbation, the one-pole tangent geometry, the projection coefficients, the residual norm, and the resulting quartic law. The boundary law, pseudo-true shift formulas, enriched scalar nulls, and scalar Monte Carlo are omitted because the multivariate analysis does not depend on them; a full treatment appears in Ref.~\cite{BiZhangCalhoun2026scalar}.

\subsection{A1. Scalar model and exact spectrum}

Consider
\begin{equation}
X_{t+1}=aX_t+\lambda F_t+\epsilon_{t+1},
\qquad
F_{t+1}=bF_t+\eta_{t+1},
\label{eqA:model}
\end{equation}
with $|a|<1$, $|b|<1$, independent Gaussian noises, and only $X_t$ observed. The null one-pole spectrum is
\begin{equation}
S_0(\omega)=\frac{\sigma_\epsilon^2}{\Pfun_a(\omega)},
\qquad
\Pfun_c(\omega)=1+c^2-2c\cos\omega,
\label{eqA:null}
\end{equation}
whereas the exact observed spectrum is
\begin{equation}
S_{\mathrm{true}}(\omega)
=
\frac{\sigma_\epsilon^2}{\Pfun_a(\omega)}
+
\frac{\lambda^2\sigma_\eta^2}{\Pfun_a(\omega)\Pfun_b(\omega)}
=
S_0(\omega)\bigl(1+\lambda^2 h(\omega)\bigr),
\label{eqA:true}
\end{equation}
with
\begin{equation}
h(\omega)=\frac{\sigma_\eta^2}{\sigma_\epsilon^2\Pfun_b(\omega)}.
\label{eqA:h}
\end{equation}
This is the standard superposition of two linearly filtered white-noise sources. The main text reuses this structure channelwise for the inherited auto terms.

\subsection{A2. Tangent space of the one-pole manifold}

The relative one-pole null family is
\[
\frac{S_{\mathrm{null}}(\omega;\tilde a,\tilde\sigma^2)}{S_0(\omega)}
=
1+u\tilde e_1(\omega)+v\tilde e_2(\omega)+O(u^2+v^2+uv),
\]
where
\begin{equation}
\tilde e_1(\omega)=1,
\qquad
\tilde e_2(\omega)=\frac{2(\cos\omega-a)}{\Pfun_a(\omega)}.
\label{eqA:basis}
\end{equation}
Thus the scalar one-pole tangent space is $\mathcal T=\mathrm{span}\{\tilde e_1,\tilde e_2\}$.

The orthogonality follows from the Jensen identity
\begin{equation}
\frac{1}{2\pi}\int_{-\pi}^{\pi}\log \Pfun_a(\omega)\,d\omega=0
\qquad (|a|<1).
\label{eqA:jensen}
\end{equation}
Differentiating with respect to $a$ gives
\[
0
=
\frac{1}{2\pi}\int_{-\pi}^{\pi}\frac{2(a-\cos\omega)}{\Pfun_a(\omega)}\,d\omega
=
-\langle \tilde e_1,\tilde e_2\rangle,
\]
so
\begin{equation}
\langle \tilde e_1,\tilde e_2\rangle=0.
\label{eqA:orth}
\end{equation}
Moreover,
\begin{equation}
\|\tilde e_1\|_{L^2}^2=1,
\qquad
\|\tilde e_2\|_{L^2}^2=\frac{2}{1-a^2}.
\label{eqA:norms}
\end{equation}
To see the second identity directly, differentiate Eq.~\eqref{eqA:jensen} twice:
\[
0
=
\frac{1}{2\pi}\int_{-\pi}^{\pi}
\left[
\frac{2}{\Pfun_a(\omega)}
-\frac{4(a-\cos\omega)^2}{\Pfun_a(\omega)^2}
\right]d\omega.
\]
Since $\tilde e_2(\omega)=-2(a-\cos\omega)/\Pfun_a(\omega)$, this yields
\[
\|\tilde e_2\|_{L^2}^2
=
\frac{1}{2\pi}\int_{-\pi}^{\pi}\frac{2\,d\omega}{\Pfun_a(\omega)}
=
\frac{2}{1-a^2},
\]
where the last step uses the standard Poisson-kernel integral for $|a|<1$.
The main text uses this two-dimensional tangent geometry independently in each observed channel.

\subsection{A3. Projection coefficients and residual norm}

With the normalized $L^2$ inner product,
\[
\langle f,g\rangle
:=
\frac{1}{2\pi}\int_{-\pi}^{\pi}f(\omega)g(\omega)\,d\omega,
\]
the scalar perturbation coefficients are
\begin{equation}
\langle h,\tilde e_1\rangle
=
\frac{\sigma_\eta^2}{\sigma_\epsilon^2}\frac{1}{1-b^2},
\label{eqA:he1}
\end{equation}
\begin{equation}
\langle h,\tilde e_2\rangle
=
\frac{\sigma_\eta^2}{\sigma_\epsilon^2}
\frac{2b}{(1-ab)(1-b^2)}.
\label{eqA:he2}
\end{equation}
For completeness, set $z=e^{i\omega}$ so that $d\omega=dz/(iz)$. Then
\[
\langle h,\tilde e_2\rangle
=
\frac{\sigma_\eta^2}{\sigma_\epsilon^2}
\frac{1}{2\pi i}
\oint_{|z|=1}
\frac{z^2-2az+1}{(z-a)(1-az)(z-b)(1-bz)}\,dz.
\]
The poles inside the unit circle are at $z=a$ and $z=b$, with residues
\[
\operatorname*{Res}_{z=a}
=
\frac{1}{(a-b)(1-ab)},
\qquad
\operatorname*{Res}_{z=b}
=
\frac{1-2ab+b^2}{(b-a)(1-ab)(1-b^2)}.
\]
Their sum simplifies to $2b/[(1-ab)(1-b^2)]$, proving Eq.~\eqref{eqA:he2}. The squared norm of $h$ is
\begin{equation}
\|h\|_{L^2}^2
=
\frac{\sigma_\eta^4}{\sigma_\epsilon^4}
\frac{1+b^2}{(1-b^2)^3}.
\label{eqA:hnorm}
\end{equation}
Therefore the orthogonal residual $R=h-\Pi_{\mathcal T}h$ satisfies
\begin{equation}
\|R\|_{L^2}^2
=
\|h\|_{L^2}^2
-\frac{\langle h,\tilde e_1\rangle^2}{\|\tilde e_1\|^2}
-\frac{\langle h,\tilde e_2\rangle^2}{\|\tilde e_2\|^2}
=
\frac{\sigma_\eta^4}{\sigma_\epsilon^4}
\frac{2b^2(a-b)^2}{(1-b^2)^3(1-ab)^2}.
\label{eqA:Rnorm}
\end{equation}
This expression is exactly the source of the inherited auto coefficients in the main text.

\subsection{A4. Quartic law and zero set}

The scalar local Whittle/Kullback--Leibler minimum is
\begin{equation}
D_{\mathrm{loc}}^{\mathrm{scalar}}(\lambda)
=
\frac{\lambda^4}{4}\|R\|_{L^2}^2+O(\lambda^6)
=
C_{\mathrm{scalar}}\lambda^4+O(\lambda^6),
\label{eqA:quartic}
\end{equation}
with
\begin{equation}
C_{\mathrm{scalar}}
=
\frac{\sigma_\eta^4}{2\sigma_\epsilon^4}
\frac{b^2(a-b)^2}{(1-b^2)^3(1-ab)^2}.
\label{eqA:Cscalar}
\end{equation}
Hence
\begin{equation}
C_{\mathrm{scalar}}=0
\iff
(a=b)\ \text{or}\ (b=0).
\label{eqA:zeroset}
\end{equation}
The main text shows that the diagonal-null cross block contributes a strictly positive quartic coefficient even at exact coalescence $a_1=a_2=b$, thereby removing the single-channel detectability singularity.

\section{Exact Multivariate Spectrum and Cross-Spectrum Lemmas}

The compact main-text form in Eq.~\eqref{eq:true-matrix} expands componentwise to
\begin{equation}
\Strue(\omega)=
\begin{pmatrix}
\dfrac{\sigma_{\epsilon_1}^2}{\Pfun_{a_1}(\omega)}+\dfrac{\lambda^2u_1^2\sigma_\eta^2}{\Pfun_{a_1}(\omega)\Pfun_b(\omega)}
&
\dfrac{\lambda^2u_1u_2\sigma_\eta^2}{(1-a_1e^{-i\omega})(1-a_2e^{i\omega})\Pfun_b(\omega)}
\\[1.2em]
\dfrac{\lambda^2u_1u_2\sigma_\eta^2}{(1-a_1e^{i\omega})(1-a_2e^{-i\omega})\Pfun_b(\omega)}
&
\dfrac{\sigma_{\epsilon_2}^2}{\Pfun_{a_2}(\omega)}+\dfrac{\lambda^2u_2^2\sigma_\eta^2}{\Pfun_{a_2}(\omega)\Pfun_b(\omega)}
\end{pmatrix}.
\label{eqB:true-matrix-components}
\end{equation}
The basic cross-spectrum identity is
\begin{equation}
S_{12}^{\mathrm{true}}(\omega)
=
\frac{\lambda^2u_1u_2\sigma_\eta^2}
{(1-a_1e^{-i\omega})(1-a_2e^{i\omega})\Pfun_b(\omega)}.
\label{eqB:S12}
\end{equation}
Taking the modulus square gives
\begin{equation}
|S_{12}^{\mathrm{true}}(\omega)|^2
=
\frac{\lambda^4u_1^2u_2^2\sigma_\eta^4}
{\Pfun_{a_1}(\omega)\Pfun_{a_2}(\omega)\Pfun_b(\omega)^2}.
\label{eqB:crossmod}
\end{equation}
These are the spectral lemmas used by the main-text cross theorem.

\paragraph{CSM from finite-rank hidden sectors.} We now show that the channel-separable multiplicative form Eq.~\eqref{eq:csm-general} of the main text is not an ad hoc parametrisation but is forced by three standard physical settings once one assumes (a)~each observed channel $X^{(i)}$ responds to its own innovation $\epsilon_t^{(i)}$ through a causal stable scalar filter $H_i(z)$, (b)~a finite-dimensional latent process $Z_t\in\mathbb{R}^r$ drives both channels through a fixed loading matrix $L\in\mathbb{C}^{2\times r}$ at the same dynamical point as the innovations, and (c)~observed and latent innovations are mutually independent.

\begin{proposition}[Channel-separable multiplicative form from finite-rank hidden sectors]
\label{app-prop:csm-physical}
Under assumptions (a)--(c), the observed two-channel spectral density is
\begin{equation}
\Strue(\omega)=\mathbf H(\omega)\bigl[\Sigma_\epsilon+K_F(\omega)\bigr]\mathbf H(\omega)^{\ast},
\qquad
K_F(\omega)=L\,S_Z(\omega)\,L^{\ast},
\label{eqB:csm-derived}
\end{equation}
with $\mathbf H(\omega)=\diag(H_1(e^{-i\omega}),H_2(e^{-i\omega}))$, $\Sigma_\epsilon=\diag(\sigma_{\epsilon_1}^2,\sigma_{\epsilon_2}^2)$, $L\in\mathbb{C}^{2\times r}$ the loading matrix, and $S_Z(\omega)$ the $r\times r$ Hermitian positive-semidefinite spectral density of the latent process $Z_t$.
\end{proposition}

\emph{Proof.}---Under assumption~(a), each observed channel has the time-domain equation of motion $X_t^{(i)}=H_i(q)\bigl[\epsilon_t^{(i)}+(LZ)_t^{(i)}\bigr]$, where $q$ is the backward shift $qX_t=X_{t-1}$ and $(LZ)_t^{(i)}:=\sum_{k=1}^{r}L_{ik}Z_t^{(k)}$ collects the latent drive into channel $i$ through row $i$ of the loading matrix. Because the filter $H_i$ is the same whether it acts on the innovation or on the latent drive entering channel $i$, we can write the joint Fourier transform as
\[
\widetilde X^{(i)}(\omega)=H_i(e^{-i\omega})\Bigl[\widetilde\epsilon^{(i)}(\omega)+\sum_{k=1}^{r}L_{ik}\widetilde Z^{(k)}(\omega)\Bigr].
\]
The observed spectral matrix is then $\Strue(\omega)=\mathbb E[\widetilde{\mathbf X}(\omega)\widetilde{\mathbf X}(\omega)^{\ast}]/T$, where $\widetilde{\mathbf X}=\mathbf H\cdot(\widetilde\epsilon+L\widetilde Z)$ pointwise in $\omega$. Expanding, and using assumption~(c) that $\widetilde\epsilon$ and $\widetilde Z$ are uncorrelated,
\[
\Strue(\omega)=\mathbf H(\omega)\,\mathbb E\bigl[(\widetilde\epsilon+L\widetilde Z)(\widetilde\epsilon+L\widetilde Z)^{\ast}\bigr]\mathbf H(\omega)^{\ast}=\mathbf H(\omega)\bigl[\Sigma_\epsilon+L\,S_Z(\omega)\,L^{\ast}\bigr]\mathbf H(\omega)^{\ast}.
\]
Setting $K_F(\omega):=LS_Z(\omega)L^{\ast}$ yields Eq.~\eqref{eqB:csm-derived}. Hermitian positive-semidefiniteness of $K_F$ is inherited from $S_Z$: for any $v\in\mathbb{C}^2$, $v^{\ast}K_F v=(L^{\ast}v)^{\ast}S_Z(L^{\ast}v)\ge 0$ because $S_Z$ is PSD.\hfill$\square$

\paragraph{Physical instances of CSM.} Proposition~\ref{app-prop:csm-physical} covers three canonical settings from nonequilibrium statistical physics:
(i) \emph{Mori--Zwanzig reduction with a finite-rank retained memory kernel.}---Projecting a high-dimensional Markovian stochastic dynamics onto two observed degrees of freedom yields, for generic projectors, a finite-rank memory kernel supplied by the orthogonal subspace~\cite{Zwanzig1961,Mori1965,ChorinHaldKupferman2000,GivonKupfermanStuart2004}. The retained-subspace dynamics are scalar per channel (filter $H_i$), and the memory kernel enters as a rank-$r$ common drive with spectral density $S_Z$.
(ii) \emph{Common-bath open-system coupling.}---A system with two probes coupled to a shared $r$-mode bath through fixed coupling vectors realizes exactly the CSM form: each probe has its own local response $H_i$ and the bath contributes a common drive $LZ$ with $L$ the coupling matrix and $S_Z$ the bath spectral density.
(iii) \emph{Multichannel linear response to a shared latent forcing bank.}---In a two-channel linear-response setting, any shared forcing with a finite-rank spectral covariance enters both channels through fixed loadings and is subsequently shaped by each channel's local susceptibility $H_i$; this is the usual setup for climate reductions with shared slow modes~\cite{Hasselmann1976,FrankignoulHasselmann1977,PenlandSardeshmukh1995} and for multichannel neural recordings with a common latent input~\cite{SekizawaItoOizumi2024}.

In each case the structural assumption is that the latent drive enters at the same dynamical point as the innovations (so both are shaped by the same $H_i$), and the CSM form follows. Assumption (a)--(c) therefore covers a physically broad class while preserving the algebraic structure that enables the filter-free cancellation.

\section{Multivariate Whittle/KL Decomposition and Hermitian Log-Det Expansion}

The full matrix Whittle/Kullback--Leibler divergence between two spectral densities is
\begin{equation}
\DKL(\Stwo_1\|\Stwo_2)=\frac{1}{4\pi}\int_{-\pi}^{\pi}\bigl[\operatorname{tr}(\Stwo_2^{-1}\Stwo_1)-\log\det(\Stwo_2^{-1}\Stwo_1)-2\bigr]d\omega,
\label{eq:matrix-dkl}
\end{equation}
and the cross-spectral witness of the main text is obtained by evaluating this gap against the scalar-marginal reduction: $\Dcoh:=\DKL(\Strue\|\Sprod)$. The physical content is that, under the CSM hypothesis, this gap is filter-free.

\begin{theorem}[Filter-free cross-spectral witness under CSM]
\label{app-thm:exact-coh}
Let the observed two-channel spectral density $\Strue(\omega)$ have channel-separable multiplicative (CSM) structure, Eq.~\eqref{eq:csm-general}, with diagonal filters $\mathbf H(\omega)=\diag(H_1,H_2)$, innovation covariance $\Sigma_\epsilon=\diag(\sigma_{\epsilon_1}^2,\sigma_{\epsilon_2}^2)$, and finite-rank Hermitian positive-semidefinite latent-sector spectrum $K_F(\omega)=LS_Z(\omega)L^{\ast}$ with $L\in\mathbb{C}^{2\times r}$. Let $\Sprod(\omega)=\diag(S_{11}^{\mathrm{true}}(\omega),S_{22}^{\mathrm{true}}(\omega))$ denote the scalar-marginal reduction. Then $\Dcoh=\DKL(\Strue\|\Sprod)$ satisfies
\begin{equation}
\Dcoh=-\frac{1}{4\pi}\int_{-\pi}^{\pi}\log\bigl(1-\rho^2(\omega)\bigr)\,d\omega,
\qquad
\rho^2(\omega)=\frac{|K_{12}(\omega)|^2}{(\sigma_{\epsilon_1}^2+K_{11}(\omega))(\sigma_{\epsilon_2}^2+K_{22}(\omega))},
\label{eqC:filter-free}
\end{equation}
pointwise in $\omega$. All dependence on the observed-channel filters $H_1, H_2$ cancels identically, and $\Dcoh>0$ if and only if $K_{12}(\omega)\ne 0$ on a set of positive Lebesgue measure.
\end{theorem}

\emph{Proof.}---We proceed in six explicit steps.

\smallskip
\noindent\textit{Step 1 (CSM block entries).} Under Eq.~\eqref{eq:csm-general}, the components of $\Strue$ are
\[
S_{ij}^{\mathrm{true}}(\omega)=H_i(\omega)\,[\Sigma_\epsilon+K_F(\omega)]_{ij}\,\overline{H_j(\omega)},
\]
with $H_i(\omega):=H_i(e^{-i\omega})$. Because $\Sigma_\epsilon$ is diagonal, this reads
\[
S_{12}^{\mathrm{true}}(\omega)=H_1(\omega)\,\overline{H_2(\omega)}\,K_{12}(\omega),
\qquad
S_{ii}^{\mathrm{true}}(\omega)=|H_i(\omega)|^2\bigl(\sigma_{\epsilon_i}^2+K_{ii}(\omega)\bigr).
\]
The scalar-marginal reduction then factorizes as $\Sprod=|H|^2\,\diag(\sigma_{\epsilon_1}^2+K_{11},\sigma_{\epsilon_2}^2+K_{22})$, where $|H|^2:=\diag(|H_1|^2,|H_2|^2)$.

\smallskip
\noindent\textit{Step 2 (trace identity).} Forming $\mathbf A(\omega):=\Sprod^{-1}(\omega)\Strue(\omega)$, the diagonal entries of $\mathbf A$ are identically one: $A_{ii}(\omega)=S_{ii}^{\mathrm{true}}/S_{ii}^{\mathrm{true}}=1$, independent of $\omega$. Therefore
\[
\operatorname{tr}(\Sprod^{-1}\Strue)=A_{11}+A_{22}\equiv 2,
\qquad\forall\omega\in[-\pi,\pi].
\]

\smallskip
\noindent\textit{Step 3 (determinant identity).} The off-diagonal entries of $\mathbf A$ are
\[
A_{12}(\omega)=\frac{S_{12}^{\mathrm{true}}}{S_{11}^{\mathrm{true}}}=\frac{H_1\overline{H_2}K_{12}}{|H_1|^2(\sigma_{\epsilon_1}^2+K_{11})},
\qquad
A_{21}(\omega)=\frac{\overline{S_{12}^{\mathrm{true}}}}{S_{22}^{\mathrm{true}}}=\frac{\overline{H_1}H_2\overline{K_{12}}}{|H_2|^2(\sigma_{\epsilon_2}^2+K_{22})}.
\]
The product $A_{12}A_{21}$ cancels the $H_1\overline H_1$ and $H_2\overline H_2$ factors:
\[
A_{12}A_{21}=\frac{|H_1|^2|H_2|^2\,|K_{12}|^2}{|H_1|^2|H_2|^2\,(\sigma_{\epsilon_1}^2+K_{11})(\sigma_{\epsilon_2}^2+K_{22})}=\rho^2(\omega),
\]
which is real, non-negative, and independent of $H_1,H_2$. Hence
\[
\det(\Sprod^{-1}\Strue)=1-A_{12}A_{21}=1-\rho^2(\omega),
\qquad\forall\omega.
\]

\smallskip
\noindent\textit{Step 4 (substitution into matrix KL).} Substituting steps 2 and 3 into the matrix Whittle/KL integrand of Eq.~\eqref{eq:matrix-dkl},
\[
\DKL(\Strue\|\Sprod)=\frac{1}{4\pi}\int_{-\pi}^{\pi}\bigl[2-\log(1-\rho^2(\omega))-2\bigr]d\omega=-\frac{1}{4\pi}\int_{-\pi}^{\pi}\log\bigl(1-\rho^2(\omega)\bigr)d\omega,
\]
which is the first equality in Eq.~\eqref{eqC:filter-free}. The second equality---the filter-free closed form of $\rho^2$---was already established pointwise in step 3.

\smallskip
\noindent\textit{Step 5 (Cauchy--Schwarz positivity).} The latent-sector spectrum $K_F(\omega)$ is Hermitian positive-semidefinite by construction: $K_F(\omega)=L S_Z(\omega)L^{\ast}$ with $S_Z$ PSD. For any $v\in\mathbb{C}^2$, $v^{\ast}K_F v=(L^{\ast}v)^{\ast}S_Z(L^{\ast}v)\ge 0$. In particular, the $2\times 2$ operator Cauchy--Schwarz inequality gives
\[
|K_{12}(\omega)|^2\le K_{11}(\omega)K_{22}(\omega),
\qquad\forall\omega.
\]
Adding the strictly positive innovation scales $\sigma_{\epsilon_i}^2>0$ to the diagonal entries in the denominator strictly increases it, so
\[
\rho^2(\omega)=\frac{|K_{12}(\omega)|^2}{(\sigma_{\epsilon_1}^2+K_{11})(\sigma_{\epsilon_2}^2+K_{22})}\le\frac{K_{11}K_{22}}{(\sigma_{\epsilon_1}^2+K_{11})(\sigma_{\epsilon_2}^2+K_{22})}<1.
\]
Combined with $\rho^2\ge 0$, this yields $\rho^2(\omega)\in[0,1)$ for all $\omega$, and hence $-\log(1-\rho^2)$ is real, finite, and non-negative. The integral in Eq.~\eqref{eqC:filter-free} is therefore well-defined and non-negative; it vanishes iff $\rho^2(\omega)\equiv 0$ almost everywhere, i.e., iff $K_{12}(\omega)=0$ almost everywhere.

\smallskip
\noindent\textit{Step 6 (Cholesky uniformity in rank $r$).} The argument above used only the Hermitian PSD structure of $K_F$; no rank assumption was invoked beyond the existence of a factorization $K_F=LS_ZL^{\ast}$. By the Cholesky (or more generally, spectral) decomposition, any Hermitian PSD $2\times 2$ matrix $K_F(\omega)$ of rank $r\in\{0,1,2\}$ admits such a factorization with $L\in\mathbb{C}^{2\times r}$ and $r\times r$ Hermitian PSD $S_Z(\omega)$. The theorem therefore applies uniformly across every finite rank $r$; in particular both the rank-one single-mode benchmark ($r=1$) and the full rank-two latent sector ($r=2$) are covered by the same pointwise identities in steps 1--4. \hfill$\square$

\bigskip

\paragraph{Rank-one benchmark decomposition.} The normalized matrix Whittle/Kullback--Leibler divergence is Eq.~\eqref{eq:matrix-dkl}. With
\[
\mathbf A(\omega)=\Snull(\omega;\theta)^{-1}\Strue(\omega)
=
\begin{pmatrix}
1+\delta_1(\omega) & \alpha(\omega)\\
\beta(\omega) & 1+\delta_2(\omega)
\end{pmatrix},
\]
Since $\Strue$ is Hermitian ($S_{21}^{\mathrm{true}}=\overline{S_{12}^{\mathrm{true}}}$) and $\Snull$ is diagonal real, $\alpha\beta=A_{12}A_{21}$ gives
\begin{equation}
\alpha(\omega)\beta(\omega)
=
\frac{|S_{12}^{\mathrm{true}}(\omega)|^2}
{S_{11}^{0}(\omega)S_{22}^{0}(\omega)}
\ge 0.
\label{eqC:alphabeta}
\end{equation}
The order estimates are: $\delta_i(\omega)=O(\lambda^2)$ (each diagonal entry of $\Strue$ deviates from $\Snull$ by the hidden-driver contribution $\lambda^2 u_i^2\sigma_\eta^2/[\Pfun_{a_i}\Pfun_b]$), and $\alpha(\omega)=O(\lambda^2)$ (since $S_{12}^{\mathrm{true}}=O(\lambda^2)$ while $S_{11}^{\mathrm{null}}=O(1)$), so $\alpha\beta=O(\lambda^4)$. The exact $2\times 2$ determinant is $\det\mathbf A=(1+\delta_1)(1+\delta_2)-\alpha\beta$, giving
\[
\log\det\mathbf A
=
\log\bigl[(1+\delta_1)(1+\delta_2)-\alpha\beta\bigr]
=
\log(1+\delta_1)+\log(1+\delta_2)-\alpha\beta+O(\lambda^6),
\]
where the last step uses $\log(1-x/(1+\delta_1)(1+\delta_2))=-x/(1+\delta_1)(1+\delta_2)+O(x^2)$ with $x=\alpha\beta=O(\lambda^4)$ and $(1+\delta_i)^{-1}=1+O(\lambda^2)$, so the correction is $O(\lambda^8)$. This is the decomposition mechanism behind Eq.~\eqref{eq:decomposition-structural}.

\section{Local Diagonal Branch and Absorption Boundary}

\begin{proposition}[Small-coupling decomposition and diagonal-reparametrization invariance]
\label{app-prop:decomposition}
Consider the rank-one benchmark model~\eqref{eq:true-matrix} under the diagonal parametric reduction $\Snull(\omega;\theta)$ of Eq.~\eqref{eq:strict-diag-null}, with $\theta=(\tilde a_1,\tilde\sigma_1^2,\tilde a_2,\tilde\sigma_2^2)$ and $\theta_0=(a_1,\sigma_{\epsilon_1}^2,a_2,\sigma_{\epsilon_2}^2)$. Let $\theta^{\star}(\lambda):=\arg\min_{\theta}\DKL(\Strue\|\Snull(\cdot;\theta))$ denote the local minimizer branch. Then:
(i) The minimizer branch satisfies $\theta^{\star}(\lambda)=\theta_0+O(\lambda^2)$.
(ii) $\Dloc(\lambda):=\DKL(\Strue\|\Snull(\cdot;\theta^{\star}(\lambda)))=D_{\mathrm{auto},1}^{\min}(\lambda)+D_{\mathrm{auto},2}^{\min}(\lambda)+\Scross(\lambda)+O(\lambda^6)$.
(iii) The cross contribution is invariant under diagonal reparametrization at quartic order: $\Scross(\theta^{\star}(\lambda),\lambda)=\Scross(\theta_0,\lambda)+O(\lambda^6)$.
\end{proposition}

\emph{Proof.}---We establish the three claims in turn.

\smallskip
\noindent\textit{Claim (i): $\theta^{\star}(\lambda)=\theta_0+O(\lambda^2)$.} The minimizer is defined by the score equation $\nabla_\theta\DKL(\Strue\|\Snull(\cdot;\theta))|_{\theta=\theta^\star}=0$. Expanding $\Strue=\Snull(\cdot;\theta_0)+\lambda^2 \Delta(\omega)+O(\lambda^4)$, where $\Delta(\omega)$ encodes the hidden-driver perturbation from Eq.~\eqref{eq:true-matrix}, the Whittle KL divergence admits the Taylor expansion
\[
\DKL(\Strue\|\Snull(\cdot;\theta))=\DKL(\Strue\|\Snull(\cdot;\theta_0))+\tfrac{1}{2}(\theta-\theta_0)^\top\mathcal H(\theta_0)(\theta-\theta_0)+g(\theta_0)^\top(\theta-\theta_0)+O(\|\theta-\theta_0\|^3),
\]
with gradient $g(\theta_0):=\nabla_\theta\DKL|_{\theta_0}$ and Hessian $\mathcal H(\theta_0):=\nabla^2_\theta\DKL|_{\theta_0}$. At $\theta=\theta_0$, $\Snull(\cdot;\theta_0)$ coincides with the zeroth-order part of $\Strue$, so $g(\theta_0)$ is first-order in the perturbation: $g(\theta_0)=O(\lambda^2)$ because the perturbation enters at order $\lambda^2$. The Hessian $\mathcal H(\theta_0)$ is the $O(1)$ Whittle Fisher information of the diagonal null family at $\theta_0$, which is nonsingular under the standard identifiability condition $|a_i|<1$. By the implicit function theorem,
\[
\theta^{\star}(\lambda)-\theta_0=-\mathcal H(\theta_0)^{-1}g(\theta_0)+O(\lambda^4)=O(\lambda^2),
\]
which proves claim (i).

\smallskip
\noindent\textit{Claim (ii): Decomposition into auto and cross contributions.} Substituting the log-determinant expansion derived in the preceding section of this appendix,
\[
\log\det\mathbf A(\omega)=\log(1+\delta_1)+\log(1+\delta_2)-\alpha\beta+O(\lambda^8),
\]
into the matrix-KL integrand, the first two terms give the channelwise auto contributions and the third gives the cross term:
\[
\DKL(\Strue\|\Snull(\cdot;\theta))=\tfrac{1}{4\pi}\int\bigl[(\operatorname{tr}(\mathbf A)-2)-\log(1+\delta_1)-\log(1+\delta_2)+\alpha\beta\bigr]d\omega+O(\lambda^8).
\]
Minimizing over $\theta$ acts only on the diagonal parameters and produces $D_{\mathrm{auto},1}^{\min}(\lambda)$ and $D_{\mathrm{auto},2}^{\min}(\lambda)$, one from each channel. The cross term $\Scross(\lambda):=(4\pi)^{-1}\int\alpha\beta\,d\omega$ involves only the off-diagonal $S_{12}^{\mathrm{true}}$, which has no dependence on $\theta$ through the diagonal null; $\theta$ enters $\alpha\beta$ only through the denominators $S_{ii}^0$. This proves claim (ii) with error $O(\lambda^6)$ after discarding the $O(\lambda^8)$ log-expansion remainder.

\smallskip
\noindent\textit{Claim (iii): Diagonal-reparametrization invariance at quartic order.} We estimate the $\theta$-sensitivity of $\Scross$. Because $\Scross=(4\pi)^{-1}\int\alpha\beta\,d\omega$ with
\[
\alpha\beta=\frac{|S_{12}^{\mathrm{true}}(\omega)|^2}{S_{11}^0(\omega;\theta)S_{22}^0(\omega;\theta)},
\]
and $|S_{12}^{\mathrm{true}}|^2=O(\lambda^4)$ while $S_{ii}^0(\omega;\theta)=O(1)$ with $\theta$-derivative $O(1)$ by smoothness of the AR(1) parametrization, the gradient $\nabla_\theta\Scross$ has a numerator of order $\lambda^4$ and a denominator of order one, so $\nabla_\theta\Scross(\theta_0,\lambda)=O(\lambda^4)$. Combined with claim (i), the first-order Taylor correction is
\begin{align*}
\Scross(\theta^{\star}(\lambda),\lambda)-\Scross(\theta_0,\lambda)
&=(\theta^{\star}-\theta_0)\cdot\nabla_\theta\Scross(\theta_0,\lambda)+O(\|\theta^{\star}-\theta_0\|^2\|\nabla^2_\theta\Scross\|)\\
&=O(\lambda^2)\cdot O(\lambda^4)+O(\lambda^4)\cdot O(\lambda^4)=O(\lambda^6),
\end{align*}
which establishes the claim: the cross contribution evaluated on the minimizer branch agrees with its value at $\theta_0$ up to $O(\lambda^6)$.\hfill$\square$

\paragraph{Interpretation.} The absorption boundary proved above is what drives the physical hierarchy stated in the main text: diagonal reparametrization can attenuate each channelwise auto contribution (by as much as $D_{\mathrm{auto},i}^{\min}$ at its scalar-null minimum) but cannot absorb the cross contribution without enlarging the reduced family itself. The quartic cross coefficient $C_{\mathrm{cross}}$ is therefore intrinsic to the hidden geometry, not to the fit.

\section{Cancellation Identity and Cross Coefficient Closed Form}

\begin{lemma}[Observed-filter cancellation at the null point]
\label{app-lem:cancellation}
Suppose the rank-one benchmark assumptions hold: (i) channel $i$ is governed by a causal stable scalar filter $H_i(z)$ with innovation variance $\sigma_{\epsilon_i}^2$, so $S_{ii}^0(\omega)=\sigma_{\epsilon_i}^2|H_i(e^{-i\omega})|^2$; (ii) a single shared hidden mode with spectral density $S_F(\omega)\in L^2[-\pi,\pi]$ enters each channel at the same dynamical point as the innovation through loadings $u_1,u_2$ with $u_1^2+u_2^2=1$; (iii) observed and hidden innovations are mutually independent. Then, pointwise in $\omega$,
\begin{equation}
\frac{|S_{12}^{\mathrm{true}}(\omega)|^2}{S_{11}^0(\omega)S_{22}^0(\omega)}=\frac{\lambda^4 u_1^2 u_2^2}{\sigma_{\epsilon_1}^2\sigma_{\epsilon_2}^2}\,S_F(\omega)^2,
\label{eqE:cancellation}
\end{equation}
independently of the observed-channel filters $H_1,H_2$.
\end{lemma}

\emph{Proof.}---Under assumption (ii), the cross-spectrum of the observed process is $S_{12}^{\mathrm{true}}(\omega)=\lambda^2 u_1 u_2\,H_1(e^{-i\omega})\overline{H_2(e^{-i\omega})}\,S_F(\omega)$, because both observed-channel innovations and the shared hidden mode are filtered by the same $H_i$. Taking the modulus square,
\[
|S_{12}^{\mathrm{true}}(\omega)|^2=\lambda^4 u_1^2 u_2^2\,|H_1(e^{-i\omega})|^2\,|H_2(e^{-i\omega})|^2\,S_F(\omega)^2.
\]
The diagonal null spectrum from assumption (i) gives
\[
S_{11}^0(\omega)\,S_{22}^0(\omega)=\sigma_{\epsilon_1}^2\sigma_{\epsilon_2}^2\,|H_1(e^{-i\omega})|^2\,|H_2(e^{-i\omega})|^2.
\]
Dividing the two expressions, the $|H_1|^2|H_2|^2$ factors appear identically in numerator and denominator and cancel pointwise in $\omega$, leaving Eq.~\eqref{eqE:cancellation}.\hfill$\square$

\begin{theorem}[Quartic cross law with closed-form $\mathcal I_F$]
\label{app-thm:cross}
Under the assumptions of Lemma~\ref{app-lem:cancellation}, and combined with the small-coupling decomposition~\eqref{eq:decomposition-structural}, the leading cross contribution is
\begin{equation}
\Scross(\lambda)=\frac{\lambda^4 u_1^2 u_2^2}{\sigma_{\epsilon_1}^2\sigma_{\epsilon_2}^2}\,\mathcal I_F+O(\lambda^6),\qquad\mathcal I_F:=\frac{1}{4\pi}\int_{-\pi}^{\pi}S_F(\omega)^2\,d\omega,
\label{eqE:cross-general}
\end{equation}
independent of the observed-channel dynamics. For the AR(1) hidden mode with $S_F(\omega)=\sigma_\eta^2/\Pfun_b(\omega)$ and $|b|<1$,
\begin{equation}
\mathcal I_F=\frac{\sigma_\eta^4(1+b^2)}{2(1-b^2)^3},
\qquad
C_{\mathrm{cross}}=\frac{u_1^2 u_2^2\,\sigma_\eta^4(1+b^2)}{2\sigma_{\epsilon_1}^2\sigma_{\epsilon_2}^2(1-b^2)^3}.
\label{eqE:IF-closed-form}
\end{equation}
\end{theorem}

\emph{Proof.}---Combining Eq.~\eqref{eqE:cancellation} with the small-coupling decomposition~\eqref{eq:decomposition-structural} gives
\[
\Scross(\lambda)=\frac{1}{4\pi}\int_{-\pi}^{\pi}\alpha(\omega)\beta(\omega)\,d\omega+O(\lambda^6)=\frac{\lambda^4 u_1^2 u_2^2}{\sigma_{\epsilon_1}^2\sigma_{\epsilon_2}^2}\cdot\frac{1}{4\pi}\int_{-\pi}^{\pi}S_F(\omega)^2\,d\omega+O(\lambda^6),
\]
which is Eq.~\eqref{eqE:cross-general}. It remains to evaluate $\mathcal I_F$ for the AR(1) hidden mode.

\smallskip
\noindent\textit{Contour evaluation of $\mathcal I_F$.} For the AR(1) mode, $\Pfun_b(\omega)=|1-be^{-i\omega}|^2=(1-be^{-i\omega})(1-be^{i\omega})$ and $S_F(\omega)=\sigma_\eta^2/\Pfun_b(\omega)$, so
\[
\mathcal I_F=\frac{\sigma_\eta^4}{4\pi}\int_{-\pi}^{\pi}\frac{d\omega}{\Pfun_b(\omega)^2}.
\]
Substitute $z=e^{i\omega}$, $dz=iz\,d\omega$, so that $d\omega=dz/(iz)$, $e^{-i\omega}=z^{-1}$, $e^{i\omega}=z$. The integral becomes
\[
\int_{-\pi}^{\pi}\frac{d\omega}{\Pfun_b(\omega)^2}=\oint_{|z|=1}\frac{1}{(1-bz^{-1})^2(1-bz)^2}\cdot\frac{dz}{iz}.
\]
Multiplying numerator and denominator by $z^2$ to clear $z^{-1}$,
\[
=\oint_{|z|=1}\frac{z^2}{(z-b)^2(1-bz)^2}\cdot\frac{dz}{iz}=\frac{1}{i}\oint_{|z|=1}\frac{z\,dz}{(z-b)^2(1-bz)^2}.
\]
For $|b|<1$, the poles are at $z=b$ (inside the unit circle, double pole) and $z=1/b$ (outside, double pole). Only $z=b$ contributes. The residue at the double pole $z=b$ is
\[
\operatorname{Res}_{z=b}\frac{z}{(z-b)^2(1-bz)^2}=\lim_{z\to b}\frac{d}{dz}\left[\frac{z}{(1-bz)^2}\right].
\]
Computing the derivative,
\[
\frac{d}{dz}\left[\frac{z}{(1-bz)^2}\right]=\frac{(1-bz)^2+2bz(1-bz)}{(1-bz)^4}=\frac{(1-bz)+2bz}{(1-bz)^3}=\frac{1+bz}{(1-bz)^3}.
\]
Evaluating at $z=b$,
\[
\left.\frac{1+bz}{(1-bz)^3}\right|_{z=b}=\frac{1+b^2}{(1-b^2)^3}.
\]
By the residue theorem, the contour integral equals $2\pi i$ times this residue:
\[
\frac{1}{i}\oint_{|z|=1}\frac{z\,dz}{(z-b)^2(1-bz)^2}=\frac{1}{i}\cdot 2\pi i\cdot\frac{1+b^2}{(1-b^2)^3}=\frac{2\pi(1+b^2)}{(1-b^2)^3}.
\]
Therefore
\[
\mathcal I_F=\frac{\sigma_\eta^4}{4\pi}\cdot\frac{2\pi(1+b^2)}{(1-b^2)^3}=\frac{\sigma_\eta^4(1+b^2)}{2(1-b^2)^3},
\]
which is the first equality in Eq.~\eqref{eqE:IF-closed-form}. Substituting this into Eq.~\eqref{eqE:cross-general} and reading off the $\lambda^4$ coefficient gives $C_{\mathrm{cross}}=u_1^2 u_2^2\,\sigma_\eta^4(1+b^2)/[2\sigma_{\epsilon_1}^2\sigma_{\epsilon_2}^2(1-b^2)^3]$.\hfill$\square$

\section{Scalar-to-Multivariate Inheritance of the Auto Terms}

Each observed channel inherits the scalar quartic law with the replacements
\[
a\mapsto a_i,
\qquad
\lambda\mapsto \lambda u_i,
\qquad
\sigma_\epsilon^2\mapsto \sigma_{\epsilon_i}^2.
\]
This gives
\[
C_{\mathrm{auto}}^{(i)}
=
\frac{u_i^4\sigma_\eta^4}{2\sigma_{\epsilon_i}^4}
\frac{b^2(a_i-b)^2}{(1-b^2)^3(1-a_ib)^2},
\]
and therefore the complete diagonal-null quartic coefficient
\[
C_{\mathrm{tot}}
=
C_{\mathrm{auto}}^{(1)}+C_{\mathrm{auto}}^{(2)}+C_{\mathrm{cross}}.
\]

\begin{corollary}[Coalescence singularity removal]
\label{app-cor:coalescence}
In the AR(1) rank-one benchmark at exact timescale coalescence $a_1=a_2=b$ with $u_1 u_2\ne 0$ and $|b|<1$, both channelwise auto coefficients vanish:
\[
C_{\mathrm{auto}}^{(i)}(a_i=b)=0,\qquad i=1,2.
\]
The full local residual at quartic order reduces to the cross contribution alone,
\begin{equation}
\Dloc(\lambda)=C_{\mathrm{cross}}\lambda^4+O(\lambda^6),
\label{eqF:coalescence-positive}
\end{equation}
with $C_{\mathrm{cross}}=u_1^2 u_2^2\,\sigma_\eta^4(1+b^2)/[2\sigma_{\epsilon_1}^2\sigma_{\epsilon_2}^2(1-b^2)^3]$ from Theorem~\ref{app-thm:cross}. This residual is strictly positive.
\end{corollary}

\emph{Proof.}---We verify the three claims in turn.

\smallskip
\noindent\textit{(i) Vanishing of auto coefficients.} From the inheritance formula above, $C_{\mathrm{auto}}^{(i)}$ contains the factor $(a_i-b)^2$ in the numerator. At $a_i=b$ this factor vanishes, and because the denominators $(1-b^2)^3$ and $(1-a_i b)^2=(1-b^2)^2$ are strictly positive for $|b|<1$, the full expression is zero: $C_{\mathrm{auto}}^{(i)}(a_i=b)=0$ for $i=1,2$.

\smallskip
\noindent\textit{(ii) Residual reduces to cross term.} Substituting the two vanishing auto coefficients into the full quartic law~\eqref{eq:full-quartic},
\[
\Dloc(\lambda)=\bigl(0+0+C_{\mathrm{cross}}\bigr)\lambda^4+O(\lambda^6)=C_{\mathrm{cross}}\lambda^4+O(\lambda^6),
\]
which is Eq.~\eqref{eqF:coalescence-positive}.

\smallskip
\noindent\textit{(iii) Strict positivity of $C_{\mathrm{cross}}$.} From Theorem~\ref{app-thm:cross}, for the AR(1) hidden mode,
\[
C_{\mathrm{cross}}=\frac{u_1^2 u_2^2\,\sigma_\eta^4(1+b^2)}{2\sigma_{\epsilon_1}^2\sigma_{\epsilon_2}^2(1-b^2)^3}.
\]
Each factor is strictly positive under the hypotheses: $u_1^2 u_2^2>0$ because $u_1 u_2\ne 0$; $\sigma_\eta^4>0$ by construction; $1+b^2>0$ for any real $b$; $(1-b^2)^3>0$ for $|b|<1$; and $\sigma_{\epsilon_1}^2\sigma_{\epsilon_2}^2>0$ by construction. The product of strictly positive terms is strictly positive, so $C_{\mathrm{cross}}>0$ and $\Dloc(\lambda)>0$ to leading order for all $\lambda\ne 0$. Crucially, this positivity is independent of any observed-pole configuration that preserves coalescence, because $C_{\mathrm{cross}}$ depends only on the hidden spectral density parameters $(b,\sigma_\eta^2)$, the loadings $(u_1,u_2)$, and the innovation scales $(\sigma_{\epsilon_1}^2,\sigma_{\epsilon_2}^2)$.\hfill$\square$

\section{Boundary Characterization for Enriched Nulls}

Fix the Whittle-normalized inner product
\[
\langle f,g\rangle:=\frac{1}{4\pi}\int_{-\pi}^{\pi}\overline{f(\omega)}\,g(\omega)\,d\omega,
\qquad
\|f\|^2:=\langle f,f\rangle,
\]
which matches the normalization of the matrix-Whittle/KL integral~\eqref{eq:matrix-dkl}. Define the normalized cross-spectral direction
\[
q(\omega):=\frac{S_{12}^{\mathrm{true}}(\omega)}{\sqrt{S_{11}^{0}(\omega)S_{22}^{0}(\omega)}},
\]
so that Eq.~\eqref{eqC:alphabeta} gives $\alpha(\omega)\beta(\omega)=|q(\omega)|^2$ pointwise. An enriched null family adds a finite collection of off-diagonal directions along the scalar-marginal branch; let $\mathcal E_\rho\subset L^2[-\pi,\pi]$ denote the finite-dimensional Hermitian subspace they span, and let $\Pi_{\mathcal E_\rho}$ be the orthogonal projection onto $\mathcal E_\rho$ in the inner product above.

\begin{proposition}[Enriched-null span criterion]
\label{app-prop:enriched-span}
With the notation above, the integrated cross contribution of the rank-one benchmark and its enriched-null residual are
\begin{equation}
\Scross\;=\;\|q\|^2,
\qquad
\Scrossrho\;=\;\|q-\Pi_{\mathcal E_\rho}q\|^2,
\label{eqG:span}
\end{equation}
with $0\le\Scrossrho\le\Scross$, and exact absorption $\Scrossrho=0$ iff $q\in\mathcal E_\rho$. At leading order in the coupling, $\Scross=C_{\mathrm{cross}}\lambda^4+O(\lambda^6)$, and the ratio $\Scrossrho/\Scross$ agrees with the quartic-coefficient ratio $C_{\mathrm{cross}}^{(\rho)}/C_{\mathrm{cross}}$ up to $O(\lambda^2)$.
\end{proposition}

\emph{Proof.}---By Eq.~\eqref{eqC:alphabeta} and the definition of $\Scross$ in the decomposition~\eqref{eq:decomposition-structural}, $\Scross=(4\pi)^{-1}\int\alpha(\omega)\beta(\omega)\,d\omega=(4\pi)^{-1}\int|q(\omega)|^2\,d\omega=\|q\|^2$ in the chosen inner product. When the null family is enriched to carry off-diagonal directions in $\mathcal E_\rho$, the diagonal entries of the enriched spectrum still match the scalar marginals exactly, so they leave the cross block untouched; the enriched family's action on the cross residual is therefore purely projection onto $\mathcal E_\rho$. Minimizing the residual norm over the enrichment coefficients is a standard Hilbert-space projection problem:
\[
\Scrossrho=\min_{r\in\mathcal E_\rho}\|q-r\|^2=\|q-\Pi_{\mathcal E_\rho}q\|^2.
\]
Writing $q=\Pi_{\mathcal E_\rho}q+(q-\Pi_{\mathcal E_\rho}q)$ with the two terms orthogonal, the Pythagorean identity gives $\|q\|^2=\|\Pi_{\mathcal E_\rho}q\|^2+\|q-\Pi_{\mathcal E_\rho}q\|^2$, so
\[
\Scrossrho=\|q\|^2-\|\Pi_{\mathcal E_\rho}q\|^2=\Scross-\|\Pi_{\mathcal E_\rho}q\|^2.
\]
Since $0\le\|\Pi_{\mathcal E_\rho}q\|^2\le\|q\|^2$ with the upper bound saturated iff $q\in\mathcal E_\rho$, we conclude $0\le\Scrossrho\le\Scross$, with exact absorption $\Scrossrho=0$ iff $q\in\mathcal E_\rho$.

The leading-order statement follows from Lemma~\ref{app-lem:cancellation}: $|q(\omega)|^2=\lambda^4 u_1^2 u_2^2\,S_F(\omega)^2/(\sigma_{\epsilon_1}^2\sigma_{\epsilon_2}^2)$ at the scalar-marginal branch point, so $\Scross=\lambda^4\cdot(u_1^2 u_2^2/(\sigma_{\epsilon_1}^2\sigma_{\epsilon_2}^2))\cdot\mathcal I_F=C_{\mathrm{cross}}\lambda^4$ exactly at leading order, with $O(\lambda^6)$ corrections from the higher-order decomposition in Eq.~\eqref{eq:decomposition-structural}.\hfill$\square$

\begin{corollary}[Aligned branch at exact coalescence]
\label{app-cor:aligned}
Let the enriched null family add a single off-diagonal direction $\psi(\omega)$, so $\mathcal E_\rho=\operatorname{span}\{\psi\}$. At exact AR(1) coalescence $a_1=a_2=b$ with $|b|<1$, the residual $\Scrossrho$ vanishes iff $\psi(\omega)$ is proportional to $\Pfun_b(\omega)^{-1}$ in $L^2[-\pi,\pi]$. Away from this alignment condition, $\Scrossrho>0$ strictly.
\end{corollary}

\emph{Proof.}---At AR(1) coalescence $a_1=a_2=b$, the scalar-marginal branch point gives $S_{ii}^0(\theta_0;\omega)=\sigma_{\epsilon_i}^2/\Pfun_b(\omega)$ and hence $\sqrt{S_{11}^0 S_{22}^0}=\sigma_{\epsilon_1}\sigma_{\epsilon_2}/\Pfun_b(\omega)$. The rank-one benchmark cross-spectrum is $S_{12}^{\mathrm{true}}(\omega)=\lambda^2 u_1 u_2\,H_1(\omega)\overline{H_2(\omega)}\,S_F(\omega)$, and at coalescence $H_1=H_2=(1-be^{-i\omega})^{-1}$, so $H_1\overline{H_2}=\Pfun_b(\omega)^{-1}$. Combining,
\[
S_{12}^{\mathrm{true}}(\omega)=\frac{\lambda^2 u_1 u_2\sigma_\eta^2}{\Pfun_b(\omega)^2},
\qquad
q(\omega)=\frac{S_{12}^{\mathrm{true}}(\omega)}{\sqrt{S_{11}^0(\omega)S_{22}^0(\omega)}}=\frac{\lambda^2 u_1 u_2\sigma_\eta^2}{\sigma_{\epsilon_1}\sigma_{\epsilon_2}}\cdot\Pfun_b(\omega)^{-1}.
\]
The function $q(\omega)$ is real-valued (since $\Pfun_b(\omega)$ is real) and is proportional to $\Pfun_b(\omega)^{-1}$ with a real, $\lambda$-dependent constant.

With $\mathcal E_\rho=\operatorname{span}\{\psi\}$, the orthogonal projection is
\[
\Pi_{\mathcal E_\rho}q=\frac{\langle\psi,q\rangle}{\|\psi\|^2}\psi,
\qquad
\Scrossrho=\|q\|^2-\frac{|\langle\psi,q\rangle|^2}{\|\psi\|^2}.
\]
By Cauchy--Schwarz in $L^2$, $|\langle\psi,q\rangle|^2\le\|\psi\|^2\|q\|^2$ with equality iff $\psi$ is proportional to $q$, i.e., $\psi(\omega)\propto\Pfun_b(\omega)^{-1}$. At this alignment, $\Scrossrho=0$; otherwise Cauchy--Schwarz is strict and $\Scrossrho>0$.

For the AR(1) enrichment whose added off-diagonal direction has the form $\psi(\omega)\propto[\Pfun_{a_1}(\omega)\Pfun_{a_2}(\omega)]^{-1/2}$, the alignment condition becomes
\[
[\Pfun_{a_1}(\omega)\Pfun_{a_2}(\omega)]^{-1/2}=c\,\Pfun_b(\omega)^{-1}
\]
for some constant $c$, uniformly in $\omega$. Squaring gives the identity $\Pfun_b(\omega)^2=c^{-2}\Pfun_{a_1}(\omega)\Pfun_{a_2}(\omega)$, which must hold for all $\omega$. Using the factorization $\Pfun_c(\omega)=|1-ce^{-i\omega}|^2=(1-ce^{-i\omega})(1-ce^{i\omega})$, both sides are squared moduli of polynomials in $z=e^{-i\omega}$ with all zeros strictly inside the unit disk (since $|a_i|,|b|<1$):
\[
|(1-bz)^2|^2=c^{-2}\,|(1-a_1 z)(1-a_2 z)|^2,\qquad |z|=1.
\]
By Szeg\H{o}'s factorization theorem for trigonometric polynomials \cite{Brillinger2001}, the minimum-phase analytic factor of a positive trigonometric polynomial of bounded degree is unique up to a unimodular constant; the analytic factors above must therefore be proportional:
\[
(1-bz)^2=c^{-1}\,(1-a_1 z)(1-a_2 z),\qquad z\in\mathbb{C}.
\]
Matching coefficients of $z^0$, $z^1$, and $z^2$ in this polynomial identity gives $c^{-1}=1$, $a_1+a_2=2b$, and $a_1 a_2=b^2$. Hence $(a_1,a_2)$ are both roots of $(x-b)^2=0$, so $a_1=a_2=b$. Therefore the alignment condition is satisfied \emph{only} at exact coalescence; away from this branch, Cauchy--Schwarz is strict and the residual $\Scrossrho$ is strictly positive.

(Remark on the inner product. At exact coalescence, $q(\omega)\propto\Pfun_b(\omega)^{-1}$ is real-valued, so the real and Hermitian inner products coincide on the pair $(q,\psi)$. Away from coalescence, the cross perturbation $S_{12}^{\mathrm{true}}$ is generally complex-valued through the factor $(1-a_1e^{-i\omega})^{-1}(1-a_2e^{i\omega})^{-1}$, and the Hermitian inner product is required for the Cauchy--Schwarz step to apply in full generality.)\hfill$\square$

\section{Symbolic Verification and Finite-Sample Records}

\emph{Symbolic verification.}---All scalar quartic-law identities (24 independent checks), the multivariate spectral decomposition (five identities covering the cancellation, cross coefficient, auto inheritance, determinant expansion, and diagonal-branch stability), and the benchmark thermodynamic extension collected in Appendix~K were verified symbolically at machine precision in both Mathematica and SymPy. Numerical validation across 486 parameter combinations (SymPy) and an independent grid of 180 combinations (Mathematica) yields agreement to relative error below $10^{-10}$ in every case.

\emph{Finite-sample records.}---Every $\lambda$ grid point successfully brackets the $50\%$ detection crossing, and the null-calibration false-positive rate is zero across all configurations tested. The single-channel detection threshold rises by a factor of $2.31$--$2.58$ between $\delta=0.20$ and $\delta=0.01$, while the two-channel threshold has a coefficient of variation of only $0.07$--$0.13$ across the same $\delta$ range---quantitative confirmation of the exact coalescence split.

The median threshold ratio is $r_{50}^{\mathrm{single}}\approx0.79$ and $r_{50}^{\mathrm{two}}\approx1.59$, and the two-channel ratio increases toward coalescence even though the corresponding population coefficient remains nearly flat. This indicates a finite-sample efficiency cost for cross-spectral estimation rather than a breakdown of the population theorem. To separate population signal from estimator cost, we ran a fixed-nuisance semi-oracle two-channel control and an extended $r_{50}(N)$ scan. The semi-oracle curves are markedly flatter, with coefficient-of-variation reductions from roughly $0.10$, $0.09$, $0.07$ to $0.03$, $0.03$, $0.03$ across $N=512$, $1024$, $2048$. The extended scan reaches $N=16384$ and shows $r_{50}^{\mathrm{two}}$ decreasing from about $1.67$ to $1.27$ at $\delta=0.10$ and from $1.93$ to $1.33$ at $\delta=0.02$. Both controls confirm that the absolute threshold split is robust, while the residual two-channel penalty is a finite-sample extraction effect.

Figure~\ref{fig:supp-r50} collects the baseline targeted finite-sample controls, while Supplementary Figs.~\ref{fig:supp-matched}--\ref{fig:supp-asymmetric} add matched-information fairness, exact-versus-off-coalescence semantics, a light persistence sweep, and asymmetric $a_1\neq a_2$ verification without changing the main-text theorem hierarchy. Figure~\ref{fig:supp-semantics} turns the diagonal versus aligned-enriched distinction into an explicit hypothesis-class preference experiment.

\begin{figure}[t]
  \centering
  \includegraphics[width=\textwidth]{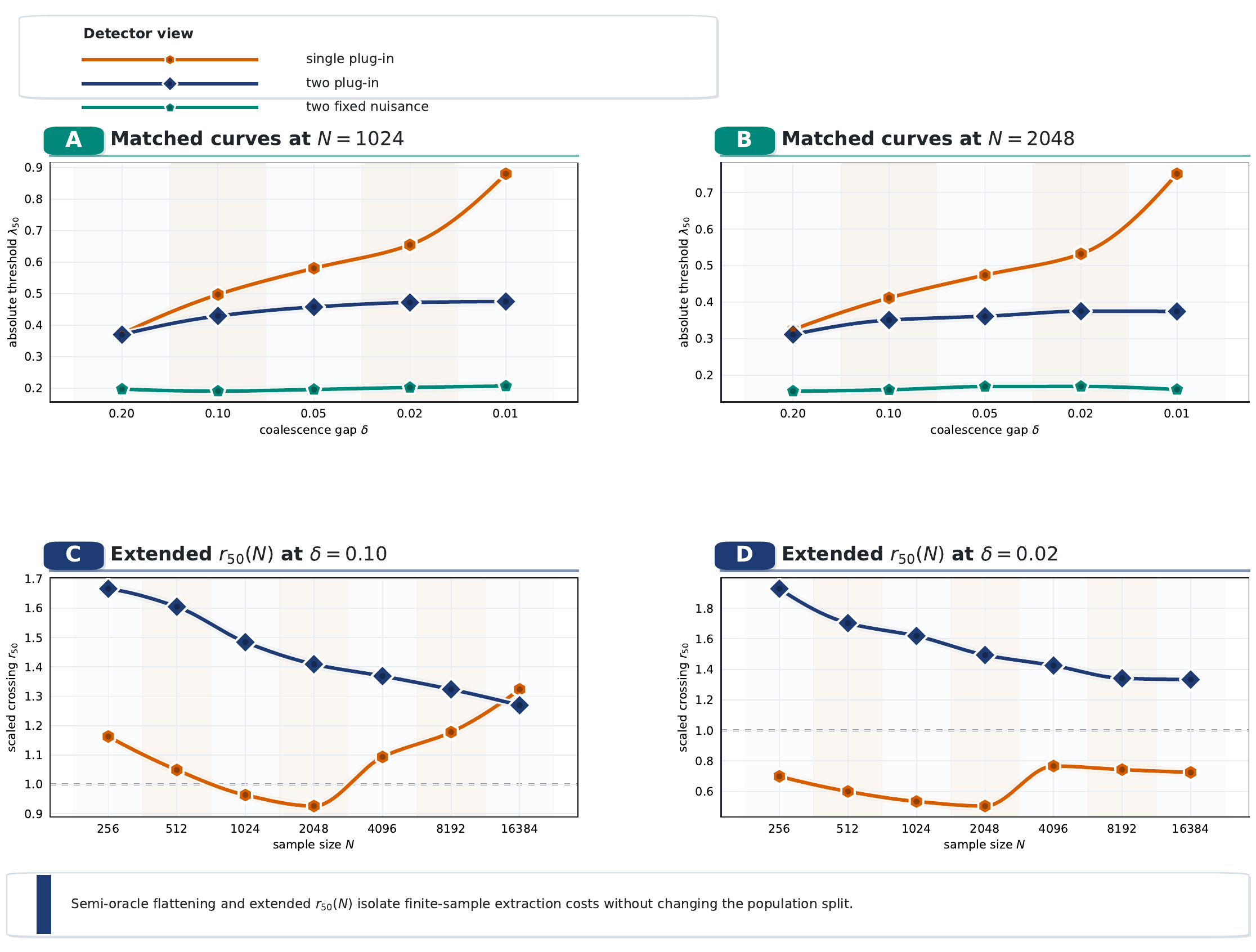}
  \caption{Targeted finite-sample controls. Panels A and B compare the single-channel plug-in curve, the two-channel plug-in curve, and a fixed-nuisance semi-oracle two-channel curve at $N=1024$ and $N=2048$, respectively. Panels C and D show the extended asymptotic trend of $r_{50}(N)$ at $\delta=0.10$ and $\delta=0.02$. The supplement therefore distinguishes exact boundedness from finite-sample extraction cost without changing the main-text theorem statements.}
  \label{fig:supp-r50}
\end{figure}

\begin{figure}[t]
  \centering
  \includegraphics[width=\textwidth]{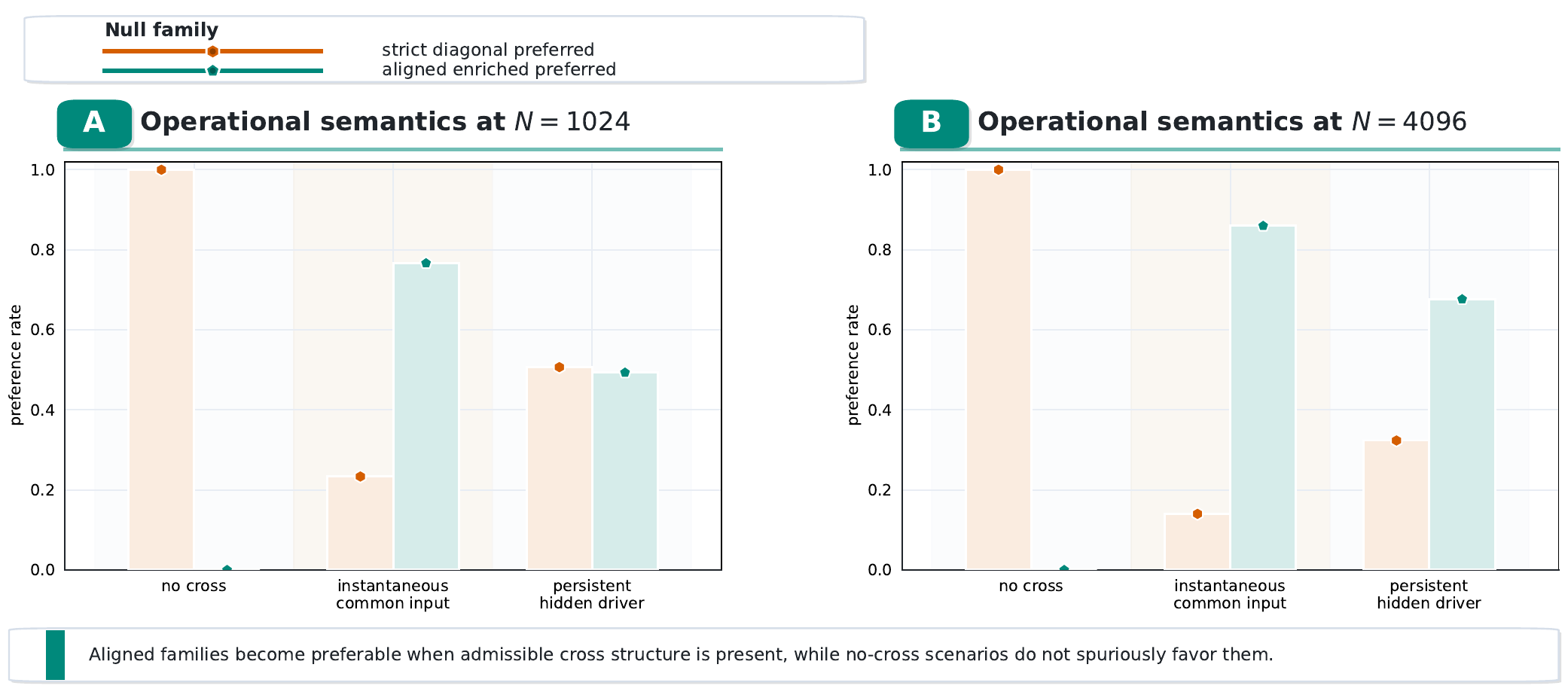}
  \caption{Hypothesis-class semantics. The two panels report, at $N=1024$ and $N=4096$, how often a one-parameter aligned cross-shape family is preferred over the diagonal family across three representative data-generating processes. The no-cross case stays at zero aligned-family preference, the instantaneous-common-input case is strongly absorbed by the aligned family, and the persistent-hidden-driver case shows intermediate-to-strong preference for the aligned family. This confirms that the two null classes answer different physical questions rather than stronger and weaker versions of the same one.}
  \label{fig:supp-semantics}
\end{figure}

\begin{figure}[t]
  \centering
  \includegraphics[width=\textwidth]{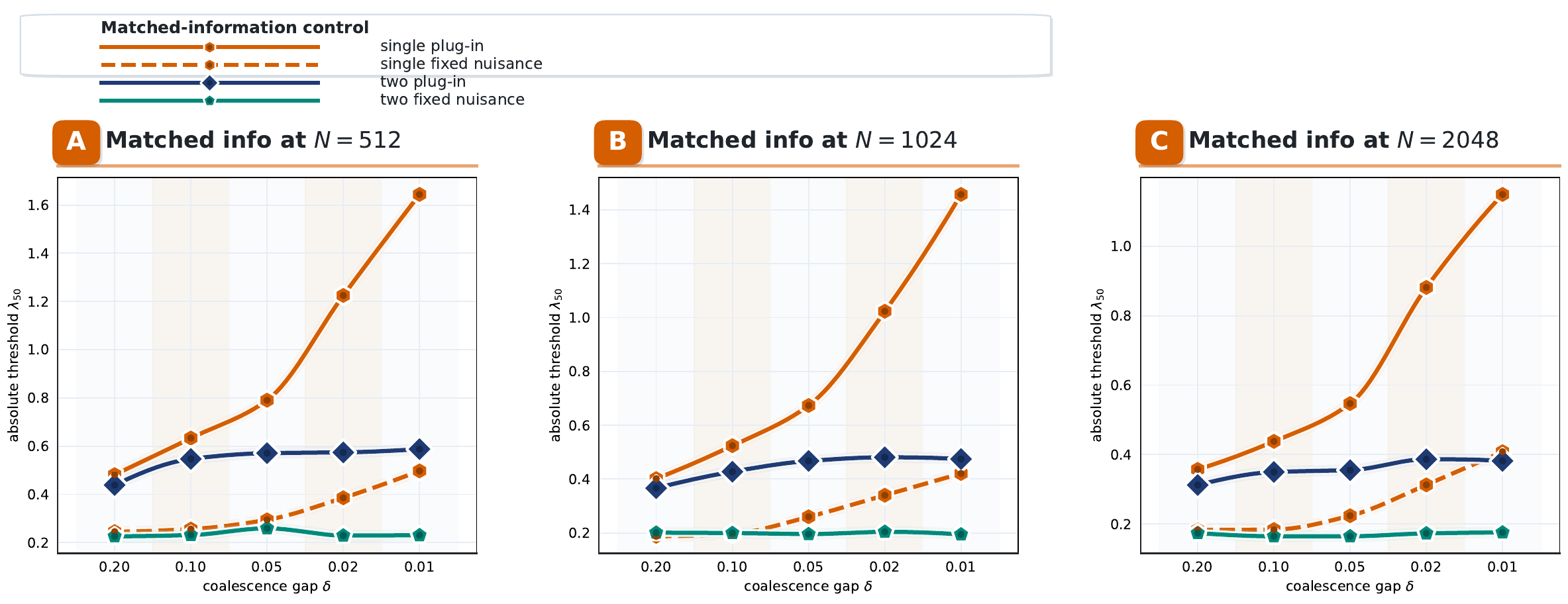}
  \caption{Matched-information control. The single-channel and two-channel reductions are compared both with plug-in nuisance estimation and with fixed nuisance values set to the true autoregressive poles and innovation scales. This diagnostic equalizes nuisance knowledge across the two reductions, separating signal geometry from extraction cost: the single-channel fixed-nuisance curves still exhibit coalescence blow-up, whereas the two-channel fixed-nuisance curves remain markedly flatter.}
  \label{fig:supp-matched}
\end{figure}

\begin{figure}[t]
  \centering
  \includegraphics[width=\textwidth]{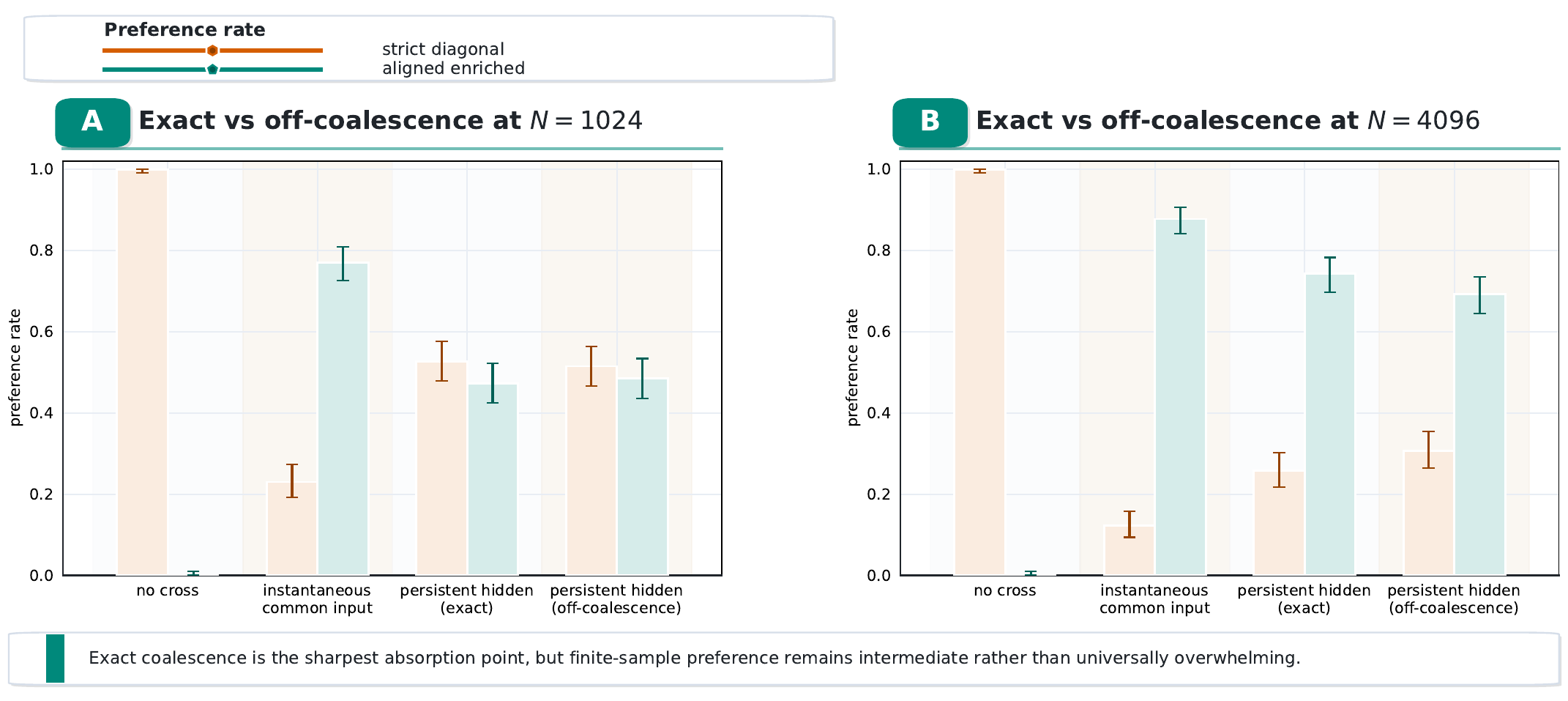}
  \caption{Refined hypothesis-class semantics with exact versus off-coalescence persistent-driver controls. The bars report diagonal-null and aligned-enriched preference rates together with Wilson $95\%$ intervals. The no-cross case remains near zero aligned-family preference, the instantaneous-common-input case is strongly aligned-family dominated, and the persistent-hidden-driver case shows intermediate-to-strong aligned-family preference in both exact and off-coalescence settings. This confirms the domain-of-validity characterization from Sec.~\ref{sec:boundary}.}
  \label{fig:supp-semanticsexact}
\end{figure}

\begin{figure}[t]
  \centering
  \includegraphics[width=\textwidth]{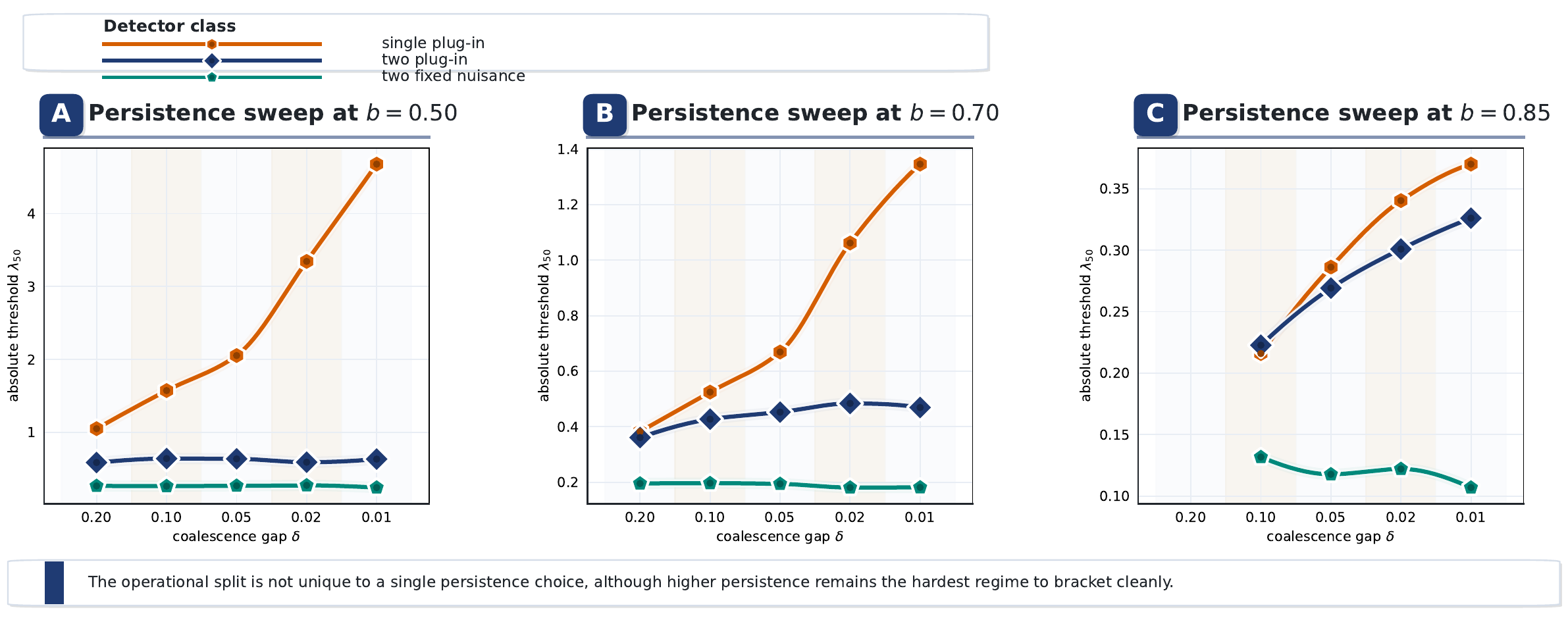}
  \caption{Persistence sweep. Across $b=0.5,0.7,0.85$ and $N=1024,2048$, the single-channel reduction continues to show a rising $\lambda_{50}(\delta)$ toward coalescence, while the two-channel curves remain flatter and the fixed-nuisance two-channel curves are flatter still. The threshold split is robust across persistence regimes, though the highest-persistence corner is the hardest to bracket cleanly at moderate sample sizes.}
  \label{fig:supp-brobust}
\end{figure}

\begin{figure}[t]
  \centering
  \includegraphics[width=\textwidth]{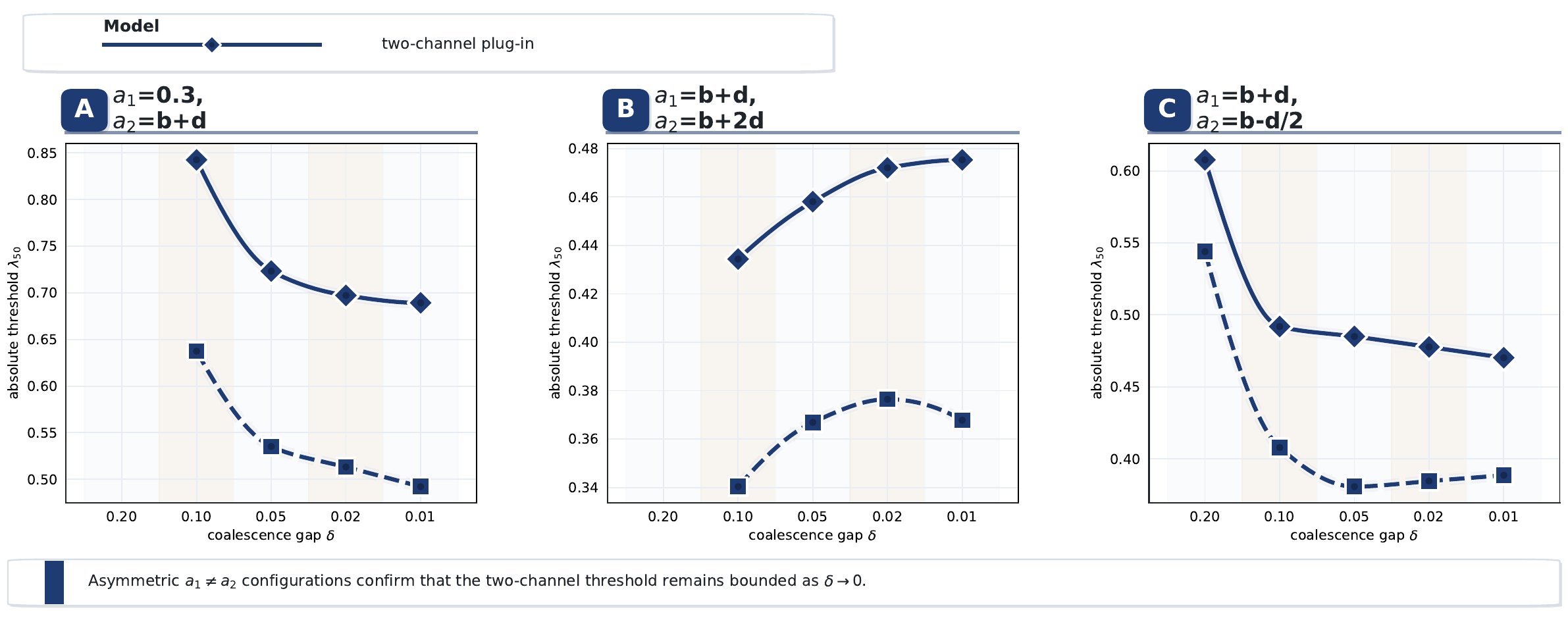}
  \caption{Asymmetric channel verification. Three configurations with $a_1\neq a_2$ confirm that the two-channel threshold $\lambda_{50}(\delta)$ remains bounded as $\delta\to0$, establishing that the symmetric $a_1=a_2$ setup used in the main-text figures is not a special case. Configurations: $a_1=b+\delta,\,a_2=b+2\delta$ (both drifting from $b$); $a_1=b+\delta,\,a_2=b-\delta/2$ (approaching from opposite sides); $a_1=0.3,\,a_2=b+\delta$ (large timescale separation). $N=1024,2048$.}
  \label{fig:supp-asymmetric}
\end{figure}

\begin{figure}[t]
  \centering
  \includegraphics[width=\textwidth]{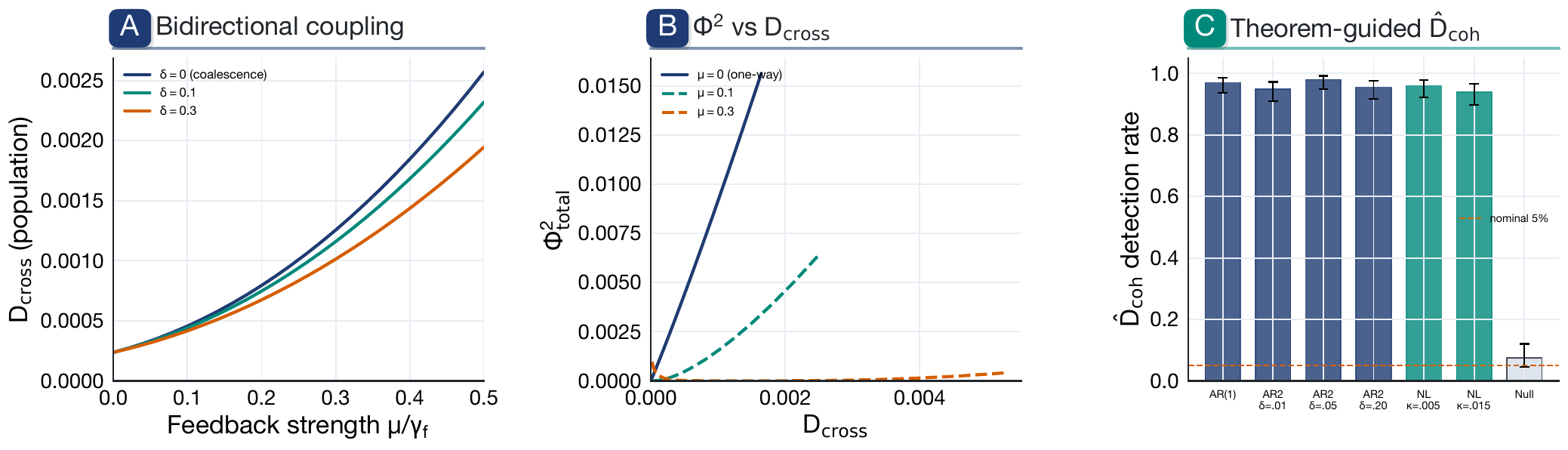}
  \caption{Structural robustness of the cross-spectral witness. Panel~A: exact $\Scross$ remains strictly positive across bidirectional feedback strengths, including at exact coalescence ($\delta=0$). Panel~B: a positive monotone association between $\EPRtot^{\,2}$ and $\Scross$ persists numerically under bidirectional coupling, deforming smoothly from the exact one-way benchmark. Panel~C: theorem-guided phase-randomization $\hat D_{\mathrm{coh}}$ test ($N=1024$, $200$ trials). AR(2): two observed poles per channel (one near coalescence, one far). NL: cubic damping $-\kappa X_i^3$. All signal scenarios achieve $92$--$97\%$ detection; the null stays at the nominal $5\%$.}
  \label{fig:robustness}
\end{figure}

\section{Robustness Experiment Protocols}
\label{app:robustness}

This appendix records the models, parameters, and test methodology for the robustness experiments in Sec.~\ref{sec:robustness} (Fig.~\ref{fig:robustness}).

\emph{Bidirectional OU model (Panels A--B).}---The continuous-time system adds symmetric feedback to Eq.~\eqref{eq:ou-model}:
\begin{equation}
dF = -\gamma_f F\,dt + \mu(u_1 X_1 + u_2 X_2)\,dt + \sqrt{2D_f}\,dW_f,
\label{eqJ:bidir}
\end{equation}
with $\gamma_f=1$, $D_1=D_2=D_f=1$, $u_1=u_2=1/\sqrt{2}$, $\lambda=0.3$. The feedback uses the same loading weights $u_1,u_2$ as the forward coupling; this symmetric choice is a simplifying assumption---asymmetric feedback weights would not change the orthogonality argument but would alter the quantitative EPR values. The drift matrix is no longer upper-triangular. The stationary covariance $\boldsymbol\Sigma$ and EPR $\EPRtot=-\operatorname{tr}(\mathbf C\mathbf D^{-1}\mathbf C\boldsymbol\Sigma^{-1})$ are computed via the $3\times3$ Lyapunov equation ($\mathbf M\boldsymbol\Sigma+\boldsymbol\Sigma\mathbf M^\top+2\mathbf D=\mathbf 0$); $\Scross$ is evaluated from the coherence integral $-(4\pi)^{-1}\int\log(1-\rho^2(\omega))\,d\omega$ on a dense frequency grid. The feedback strength $\mu/\gamma_f$ ranges from $0$ to $0.5$ across three coalescence gaps $\delta\in\{0,0.1,0.3\}$. Note that at $\mu=\lambda$ (and equal damping rates $\gamma_1=\gamma_2=\gamma_f$), the drift matrix becomes symmetric and detailed balance holds, so $\EPRtot=0$ exactly; nonetheless $\Scross$ remains strictly positive at this equilibrium point, illustrating that cross-spectral structure from shared input persists independently of thermodynamic irreversibility.

\emph{AR(2) observed dynamics (Panel C).}---Each observed channel follows an AR(2) process with poles $(p_1,p_2)$, where $p_1=b+\delta$ (near the hidden pole $b=0.7$) and $p_2=0.3$ (far pole). The AR(2) coefficients are $\varphi_1=p_1+p_2$, $\varphi_2=-p_1 p_2$ (distinct from the main-text AR(1) coefficients $a_1,a_2$). The hidden driver $F$ is AR(1) as in the main model. This tests whether a richer diagonal null---with two observed poles per channel---can absorb the cross-spectral signature.

\emph{Nonlinear cubic damping (Panel C).}---The AR(1) dynamics of Eq.~\eqref{eq:model-x} are augmented with a cubic term $-\kappa X_i(t)^3$, for $\kappa\in\{0.005,0.015\}$ at near-coalescence $\delta=0.01$. The nonlinearity is mild ($\lesssim 4\%$ correction at one standard deviation) but sufficient to test model-free detection.

\emph{Phase-randomization $\hat D_{\mathrm{coh}}$ test (Panel C).}---For each Monte Carlo trial ($N=1024$, $200$ trials), the test statistic is the theorem-guided quantity $\hat D_{\mathrm{coh}}=-(4\pi)^{-1}\sum_j\log(1-\hat\rho^2(\omega_j))\,\Delta\omega$, where $\hat\rho^2$ is computed from the band-averaged cross-periodogram (bandwidth $K=11$). The null distribution is generated by phase-randomizing $x_2$: the discrete Fourier transform of $x_2$ is multiplied by $e^{i\theta_j}$ with independent uniform $\theta_j\in[0,2\pi)$ at each interior frequency, preserving the power spectrum while destroying cross-channel phase coherence. A $p$-value is computed from $199$ surrogates per trial; detection is declared at $p<0.05$. The null control (two independent AR(1) channels, no hidden driver) yields a $6\%$ false-positive rate, consistent with the nominal level.

\section{Related Works and Scope of the Present Result}

The present result belongs to three nearby traditions. First, it sits within the reduced-order spectral and state-space analysis of stationary linear systems, Whittle likelihoods, and local information geometry \cite{Whittle1953,Brillinger2001,Lutkepohl2005,Kalman1960,KailathSayedHassibi2000}. In that language, the quartic calculation is a local statement about what a reduced one-pole null can absorb. It also connects to the physics of coarse-grained stochastic dynamics, where hidden slow modes bias entropy production estimates and activity measures \cite{Seifert2012,Esposito2012,MehlLanderSeifert2012,RoldanParrondo2010,HaradaSasa2005,Fodor2016,Nardini2017}, and to stochastic climate models where surface observables are driven by unresolved forcing \cite{Hasselmann1976,FrankignoulHasselmann1977,PenlandSardeshmukh1995}.

Second, it is closely related to the literature on cross spectra, common input, coherence, and frequency-domain dependence \cite{Granger1969,Geweke1982,Geweke1984,BresslerSeth2011}. In particular, Geweke's decomposition \cite{Geweke1982,Geweke1984} provides a general framework for separating linear dependence into auto and cross components. Our contribution is not the decomposition itself but the exact cancellation identity (Eq.~\eqref{eq:cancellation}): the observed-channel transfer-function factors $|H_1|^2,|H_2|^2$ divide out identically in the cross block before integration, yielding a coefficient that depends only on the hidden spectral density. This cancellation is a structural property of the specific null geometry and is not a consequence of the general Geweke framework; it is what makes the coalescence singularity removal possible.

Third, the paper is naturally read alongside projection-based reduced dynamics and information-geometric descriptions of model manifolds \cite{Zwanzig1961,Mori1965,ChorinHaldKupferman2000,GivonKupfermanStuart2004,AmariNagaoka2000}. In that language, the scalar dark regime is a projection singularity: the leading hidden perturbation lies in the tangent space of the reduced diagonal null and is therefore absorbed. The two-channel result changes the conclusion by changing the geometry of the retained observation class: the cross-spectral block is orthogonal to the diagonal tangent space, and its leading coefficient is governed by an exact cancellation that removes all dependence on the observed dynamics.

Fourth, the paper connects to the rapidly growing literature on entropy production estimation from partial and coarse-grained observations \cite{KawaiParrondoVandenBroeck2007,RoldanParrondo2010,BaratoSeifert2015,GingrichHorowitzPerunovEngland2016,DechantSasa2021,DechantGarnierBrunSasa2023,DianaEsposito2014,OhgaItoKolchinsky2023,HarunariEtAl2022,SekizawaItoOizumi2024,BattleEtAl2016,Maes2020}. The single-channel impossibility theorem \cite{LucenteEtAl2022,CrisantiPuglisiVillamaina2012} establishes that scalar Gaussian observations cannot detect distance from equilibrium in linear systems. Our cross-spectral analysis shows that the minimal additional observation---a second channel---qualitatively changes this picture: the cross-spectral block provides irreversibility information that is structurally inaccessible to any single-channel measure and exactly independent of the observed dynamics.

These comparisons also delimit the scope. The present result does not prove generic multivariate superiority, nor does it claim universal causal identification from cross spectra. Its precise claim is that under the diagonal null---the natural hypothesis for the absence of cross-channel dependence---the coalescence singularity is a projection artifact removed by retaining cross spectra, and that the resulting cross-spectral information certifies hidden common input---and, under one-way coupling, hidden dissipation---even when all single-channel measures are provably blind. The enriched-null analysis in the main-text boundary section characterizes the domain of validity of that statement.

\section{Benchmark-Specific Thermodynamic Extension}
\label{app:thermo}

This appendix records the one-way Ornstein--Uhlenbeck thermodynamic extension omitted from the main text to keep the Letter focused on the exact two-channel detectability theorem.

\emph{Continuous-time setup.}---Consider the OU counterpart of Eqs.~\eqref{eq:model-x}--\eqref{eq:model-f} with matching one-way hidden-driver geometry:
\begin{equation}
\begin{aligned}
dX_i &= -\gamma_i X_i\,dt + \lambda u_i F\,dt + \sqrt{2D_i}\,dW_i,\quad i=1,2,\\
dF &= -\gamma_f F\,dt + \sqrt{2D_f}\,dW_f,
\end{aligned}
\label{eq:ou-model}
\end{equation}
with drift matrix $\mathbf M$ and diffusion $\mathbf D=\diag(D_1,D_2,D_f)$. The stationary covariance $\boldsymbol\Sigma$ satisfies $\mathbf M\boldsymbol\Sigma+\boldsymbol\Sigma\mathbf M^\top+2\mathbf D=\mathbf 0$. The steady-state entropy production rate is \cite{SpinneyFord2012,GilsonEtAl2023,LandiPaternostro2021}
\begin{equation}
\EPRtot = -\operatorname{tr}\!\bigl(\mathbf C\,\mathbf D^{-1}\mathbf C\,\boldsymbol\Sigma^{-1}\bigr)\ge 0,
\label{eq:epr-full}
\end{equation}
where $\mathbf C=\mathbf M\boldsymbol\Sigma+\mathbf D$ is the antisymmetric irreversibility matrix ($\mathbf C^\top=-\mathbf C$). The discrete-time correspondence is $a_i=e^{-\gamma_i\Delta t}$, $b=e^{-\gamma_f\Delta t}$ with unit sampling interval $\Delta t=1$.

\emph{Continuous-time cancellation.}---Eq.~\eqref{eq:cancellation} applies directly with $H_i(\omega)=(\gamma_i+i\omega)^{-1}$ and $S_F(\omega)=2D_f/(\gamma_f^2+\omega^2)$:
\begin{equation}
\frac{|S_{12}^{\mathrm{ct}}(\omega)|^2}{S_{11}^{0,\mathrm{ct}}(\omega)\,S_{22}^{0,\mathrm{ct}}(\omega)}
=
\frac{\lambda^4 u_1^2 u_2^2 D_f^2}{D_1 D_2\,(\gamma_f^2+\omega^2)^2},
\label{eq:cancellation-ct}
\end{equation}
independent of $\gamma_1,\gamma_2$. The discrete-time correspondence is $a_i=e^{-\gamma_i}$, $b=e^{-\gamma_f}$, $\sigma_{\epsilon_i}^2=D_i(1-a_i^2)/\gamma_i$, $\sigma_\eta^2=D_f(1-b^2)/\gamma_f$; the coefficient ratio $\alpha_2^2/C_{\mathrm{cross}}$ in Corollary~\ref{cor:epr-dcross} below is invariant under this mapping (verified symbolically in Appendix~H).

\emph{Single-channel impossibility.}---When only one channel $X_1$ is observed, all information about the hidden driver must be extracted from the marginal scalar time series. For linear Gaussian systems, the marginal statistics are identically time-reversible \cite{LucenteEtAl2022,CrisantiPuglisiVillamaina2012}:
\begin{equation}
\EPR_{\mathrm{single}}^{\mathrm{apparent}}=0 \qquad\text{(identically, for all parameter values).}
\label{eq:epr-single}
\end{equation}

\emph{Exact EPR and the cross-spectral relationship.}---The cancellation identity (Eq.~\eqref{eq:cancellation}) holds for general observed dynamics; the following EPR results exploit the specific one-way OU structure to obtain a sharp quantitative benchmark. The one-way coupling of the model~\eqref{eq:ou-model} (the hidden mode $F$ evolves independently of the observed channels) yields an exact closed-form EPR.

\begin{proposition}[Benchmark-specific exact EPR for one-way coupled OU]
\label{thm:epr-exact}
For the system~\eqref{eq:ou-model}, the full-system entropy production rate is
\begin{equation}
\EPRtot = \alpha_2\,\lambda^2,
\label{eq:epr-exact-appK}
\end{equation}
exactly for all $\lambda$, where
\begin{equation}
\alpha_2 = \frac{u_1^2 D_f}{D_1(\gamma_1+\gamma_f)}+\frac{u_2^2 D_f}{D_2(\gamma_2+\gamma_f)}>0.
\label{eq:alpha2}
\end{equation}
\end{proposition}

\emph{Proof.}---One-way coupling makes $\mathbf M$ upper triangular, so the Lyapunov equation decouples block by block. Write $\boldsymbol\Gamma=\diag(\gamma_1,\gamma_2)$ and $\mathbf g=(u_1,u_2)^\top$. The $(x,f)$ block gives $\boldsymbol\Sigma_{xf}=\lambda(\gamma_f\mathbf I+\boldsymbol\Gamma)^{-1}\mathbf g\,D_f/\gamma_f$ exactly, making $\mathbf C_{xf}$ linear in $\lambda$ and contributing $\alpha_2\lambda^2$ to the EPR. The remaining $O(\lambda^4)$ contribution involves both the $(x,x)$-block irreversibility correction $\mathbf C_{xx}^{(2)}$ (which is antisymmetric, with entries $\propto(\gamma_j-\gamma_i)/(\gamma_i+\gamma_j)$) and the Schur-complement correction to $\boldsymbol\Sigma^{-1}$ from the off-diagonal covariance $\boldsymbol\Sigma_{xf}$. These two $O(\lambda^4)$ terms cancel exactly: one-way coupling constrains $\mathbf C_{xx}^{(2)}=\boldsymbol\Sigma_{xf}\Sigma_{ff}^{-1}\mathbf C_{xf}^\top-\mathbf C_{xf}\Sigma_{ff}^{-1}\boldsymbol\Sigma_{xf}^\top$, and the EPR contribution of this antisymmetric correction is exactly offset by the Schur complement of $\boldsymbol\Sigma_{xf}$ in $\boldsymbol\Sigma^{-1}$. Hence $\EPRtot=\alpha_2\lambda^2$ exactly (confirmed by independent symbolic computation in both SymPy and Mathematica across $486+180$ parameter combinations; Appendix~H).\hfill$\square$

\begin{corollary}[EPR--detectability relationship]
\label{cor:epr-dcross}
The full-system EPR and the cross-spectral detectability $\Scross=C_{\mathrm{cross}}\lambda^4+O(\lambda^6)$ satisfy
\begin{equation}
\EPRtot^{\,2} = \frac{\alpha_2^2}{C_{\mathrm{cross}}}\,\Scross + O(\lambda^6),
\label{eq:epr-dcross}
\end{equation}
where $C_{\mathrm{cross}}$ is the observed-dynamics-independent coefficient from Eq.~\eqref{eq:cross-general} and $\alpha_2$ is given by~\eqref{eq:alpha2}. Equivalently,
\begin{equation}
\EPRtot = \frac{\alpha_2}{\sqrt{C_{\mathrm{cross}}}}\sqrt{\Scross}+O(\lambda^4).
\label{eq:epr-sqrt-bound}
\end{equation}
In particular, for the present model class (one-way coupling), $\Scross>0$ implies $\EPRtot>0$: a strictly positive cross-spectral residual under the diagonal null witnesses full-system entropy production, even when all single-channel EPR estimators return zero.
\end{corollary}

\noindent\emph{Identification hierarchy.}---The cross-spectral residual $\Scross>0$ certifies that a hidden common driver is present (structural detection); the additional one-way coupling condition ensures that this driver entails dissipation (thermodynamic identification). The one-way condition is therefore not a technical limitation but a physical identification condition for this benchmark extension.

\begin{figure}[t]
  \centering
  \includegraphics[width=\textwidth]{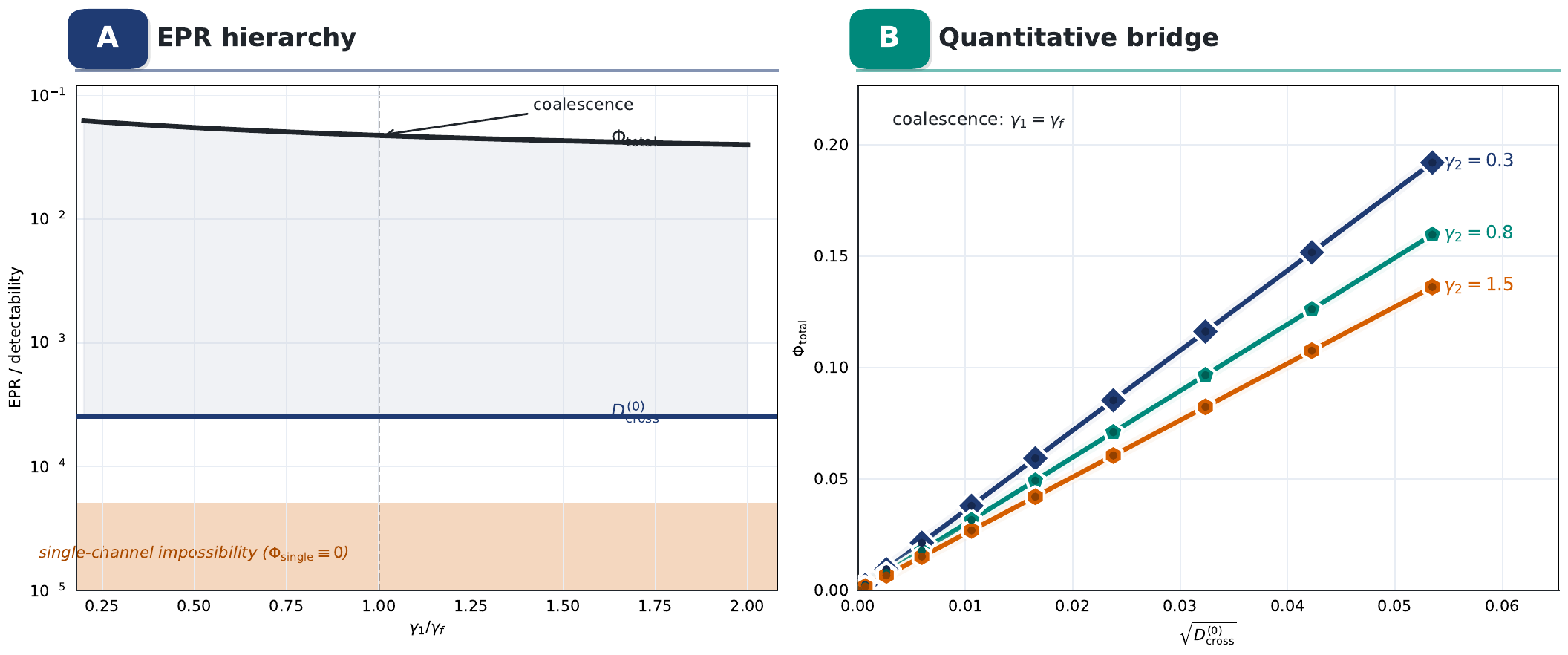}
  \caption{\textbf{Hidden entropy production through the cross-spectral witness.} \textbf{A,} The full-system steady-state entropy production rate $\EPRtot=\alpha_2\lambda^2$ (black) is strictly positive across the entire range of observed-to-hidden timescale ratios $\gamma_1/\gamma_f$, while the cross-spectral witness $\Scross=C_{\mathrm{cross}}\lambda^4$ (blue, horizontal) remains finite and independent of $\gamma_1$. Single-channel entropy production estimators return zero identically (Lucente--Crisanti)---all dissipative content of the system is carried by the cross-spectral witness. \textbf{B,} At exact timescale coalescence $\gamma_1=\gamma_f$, the leading-order relation $\EPRtot=(\alpha_2/\sqrt{C_{\mathrm{cross}}})\sqrt{\Scross}+O(\lambda^4)$ appears as a family of straight rays for representative values of $\gamma_2$, with slope set by $\alpha_2/\sqrt{C_{\mathrm{cross}}}$. The two panels together establish the identification hierarchy from cross-spectral detection to thermodynamic dissipation under one-way coupling.}
  \label{fig:epr-hierarchy}
\end{figure}

\section{Finite-Rank Latent-Sector Verification (Route~1)}
\label{app:finite-rank}

We record a direct numerical verification of Eq.~\eqref{eq:rho-csm} beyond the rank-one benchmark, using genuinely rank-$2$ latent sectors with both same-sign and mixed-sign channel loadings. The reproducible driver script is \texttt{code/route1\_finite\_rank\_feasibility.py}; source CSVs are \texttt{code/results/route1\_finite\_rank\_population.csv} and \texttt{code/results/route1\_finite\_rank\_detection.csv}.

\paragraph{Population-level filter invariance.} For each of three latent scenarios---(i) rank-$2$ shared-sector with loadings $L=\bigl(\begin{smallmatrix}1&0.5\\0.5&1\end{smallmatrix}\bigr)$, (ii) rank-$2$ mixed-sign sector with loadings $L=\bigl(\begin{smallmatrix}1&-0.5\\0.5&1\end{smallmatrix}\bigr)$, and (iii) a diagonal-only latent sector with no cross coupling (null control)---we computed the exact $\Dcoh$ under four distinct channel-separable observed filters: coalescent AR(1) ($a_1=a_2=0.7$), far-apart AR(1) ($a_1=0.2,a_2=0.8$), mixed AR(2), and mixed AR(3). The numerical spread of $\Dcoh$ across observed-filter classes was $0.00\times 10^0$ for every scenario, and the pointwise maximum discrepancy between the raw $\rho^2(\omega)$ and the filter-free closed form~\eqref{eq:rho-csm} was $1.11\times 10^{-16}$ (machine precision). For the no-cross-coupling null scenario, $\Dcoh=0$ exactly for every filter class. This confirms the operator Cauchy--Schwarz closure and the filter-free cancellation of Eq.~\eqref{eq:rho-csm} for every finite-rank Hermitian PSD latent sector tested.

\paragraph{Finite-sample detection under misspecification.} With sample size $N=1024$, burn-in $250$, and $48$ Monte Carlo trials per scenario, we compared three detectors at supercritical coupling $\lambda/\lambda_c=1.10$: (a) the theorem-guided $\hat D_{\mathrm{coh}}$ with phase-randomized surrogates ($49$ replicates), (b) a rank-one parametric plugin fit to a misspecified (rank-$2$) truth, and (c) a single-channel irreversibility baseline. On rank-$2$ shared-sector truths, detection rates were $32/48=66.7\%$ (theorem-guided), $7/48=14.6\%$ (rank-one plugin), and $0/48=0\%$ (single-channel). On rank-$2$ mixed-sign truths they were $29/48=60.4\%$, $1/48=2.1\%$, and $0/48=0\%$, respectively. Under the no-cross-coupling null control, the theorem-guided detector yielded $2/48=4.2\%$ false positives, consistent with the nominal $\alpha=0.05$ rate. The rank-one parametric fit collapses when the latent truth is genuinely rank-$2$, while the theorem-guided statistic remains powerful because it depends only on the coherence ratio~\eqref{eq:rho-csm} and inherits the filter-free invariance of Eq.~\eqref{eq:rho-csm}.

These checks show that the main theorem's finite-rank scope is not only an algebraic extension of the rank-one benchmark but also a finite-sample upgrade: when the hidden truth is multi-mode, the theorem-guided detector remains the appropriate statistic, and rank-one plugin fits are correspondingly misspecified.

\end{document}